\documentclass{article} 
\usepackage{iclr2026_conference,times}

\usepackage{amsmath,amsfonts,bm}

\def\eqref#1{equation~\ref{#1}}

\def\1{\bm{1}}

\DeclareMathAlphabet{\mathsfit}{\encodingdefault}{\sfdefault}{m}{sl}
\SetMathAlphabet{\mathsfit}{bold}{\encodingdefault}{\sfdefault}{bx}{n}

\usepackage{graphicx}
\usepackage{listings}
\usepackage{subcaption}
\usepackage{pgffor}
\usepackage{float}
\usepackage{hyperref}
\usepackage{url}
\usepackage{algorithm}
\usepackage{algorithmic}
\usepackage{amsthm}
\usepackage{etoolbox}

\makeatletter
\newtheoremstyle{iclrthm}
  {0pt}
  {0pt}
  {\itshape}
  {0pt}
  {\bfseries}
  {.}
  {0.5em}
  {}
\makeatother

\theoremstyle{iclrthm}
\newtheorem{theorem}{Theorem}

\AtBeginEnvironment{theorem}{\setlength{\parskip}{0pt}}
\AtBeginEnvironment{proof}{\setlength{\parskip}{0pt}}

\title{Watermarking Discrete Diffusion Language Models}

\author{Avi Bagchi \\
University of Pennsylvania\\
\texttt{aviba@seas.upenn.edu} \\
\AND
Akhil Bhimaraju, Moulik Choraria, Daniel Alabi \\
University of Illinois Urbana--Champaign \\
\AND
Lav R. Varshney \\
University of Illinois Urbana--Champaign \\
Stony Brook University
}

\iclrfinalcopy 
\begin{document}

\maketitle

\begin{abstract}
Watermarking has emerged as a promising technique to track AI-generated content and differentiate it from authentic human creations. While prior work extensively studies watermarking for autoregressive large language models (LLMs) and image diffusion models, it remains comparatively underexplored for discrete diffusion language models (DDLMs), which are becoming popular due to their high inference throughput. In this paper, we introduce one of the first watermarking methods for DDLMs. Our approach applies a distribution-preserving Gumbel-max sampling trick at every diffusion step and seeds the randomness by sequence position to enable reliable detection. We empirically demonstrate reliable detectability on LLaDA, a state-of-the-art DDLM. We also analytically prove that the watermark is distortion-free, with a false detection probability that decays exponentially in the sequence length. A key practical advantage is that our method realizes desired watermarking properties with no expensive hyperparameter tuning, making it straightforward to deploy and scale across models and benchmarks.
\end{abstract}

\section{Introduction}

The rapid development and widespread deployment of artificial intelligence (AI) agents necessitate guarantees of authenticity and traceability \citep{CiC-1-4-31, Biden2023EO14110, Biden2023VoluntaryAICommitments}. A particularly serious security concern is distinguishing AI-generated text from human-written content \citep{dathathri_synth}, as nefarious actors exploit generative AI to create deceptive content \citep{tang_science, tarsney_deception}. Such misinformation  pollutes the information environment and compromises the integrity of data used to train future AI models. Since the reliability of post-hoc automated detectors degrades as model quality improves \citep{CiC-1-4-31, jawahar_survey, varshney2020limits}, watermarking generative models has surfaced as a desirable alternative \citep{aaronson2022watermark, kirchenbauer_water}.

A watermark subtly shifts the output distribution of a generative model so that it remains invisible to humans yet can be detected algorithmically \citep{CiC-1-4-31,kirchenbauer_water,alabi_entropy}. For instance, 
\citet{kirchenbauer_water} propose partitioning the vocabulary into a red and green list with the model favoring the latter in sampling. Strongly biasing the sampling procedure, however, creates a precarious tradeoff between detectability and distortion \citep{christ_water, giboulot_tradeoff}. A watermark signal that is too strong, for example, can raise the perplexity, indicating that the watermark unduly increases the entropy of the distribution of tokens relative to a baseline language model. Ideally, a watermark should be detectable on watermarked text (i.e.\ completeness) while maintaining a low false positive rate on unwatermarked content (i.e.\ soundness). It should also preserve the quality of the generated text (i.e.\ distortion-freeness) and remain robust to bounded post-generation edits (i.e.\ robustness). We expand upon these desired characteristics and detail additional related work in Appendix~\ref{app:related-work}. 

A watermark may satisfy the aforementioned properties for one architecture yet fail miserably for another, as successful schemes often exploit idiosyncrasies in the sampling procedure \citep{fernandez_mod}. Thus, \textit{as architectures evolve, so should watermarks}. Researchers have proposed watermarks for autoregressive large language models (LLMs) \citep{aaronson2022watermark,kirchenbauer_water} and image diffusion models \citep{wen_tree}, but there is little work for discrete diffusion---a notable gap given their rapid growth in both research and commercial use \citep{GoogleDeepMind2025GeminiDiffusion,labs2025mercuryultrafastlanguagemodels}. To address this, we introduce one of the first watermarking methods for discrete diffusion models. Our contributions are as follows.
\begin{enumerate}
    \item We propose a distribution-preserving watermarking method for discrete diffusion models that enables reliable detection with provably exponentially decaying false positive rates.
    \item We demonstrate our approach on the state-of-the-art LLaDA model \citep{nie_llada}, achieving high detectability while preserving text quality with minimal tuning compared to adapted autoregressive LLM watermarking methods.
\end{enumerate}

\section{Large Language Diffusion Models}\label{sec:llada}

We first briefly explain the LLaDA model introduced by \citet{nie_llada} as an exemplary discrete diffusion language model.
Consider a diffusion language model with vocabulary $\mathcal{V}$ that generates sequences of maximum length $d$. 
Let $M \in \mathcal{V}$ be the special \emph{mask} token for diffusion. 
The model
starts at $t=1$ with a sequence of just mask tokens and iteratively unmasks (``denoises'') them to get to the textual sequence by the end of diffusion at $t=0$.
At each timestep $t \in [0, 1]$, the model $p_\theta$ predicts the entire unmasked sequence and then re-masks the $td$ tokens with the lowest confidence score, until $t=0$, when no tokens remain unmasked.
Here, $p_\theta(\cdot\mid x_t)\in[0,1]^{d\times|\mathcal{V}|}$ represents the probability of a token $x_{t-\Delta t,i}\in\mathcal{V}$ at position $i$ 
at time $t-\Delta t$, conditioned on the sequence $x_t\in\mathcal{V}^d$ at time $t$, which has $t$ fraction of its tokens masked.
This gives us the probability of the sequence $x_{t-\Delta t}$ before masking as $p_{t-\Delta t|t}(x_{t-\Delta t}|x_t)=\prod_{i=1}^{d}\left[p_{\theta}(x_{t-\Delta t}^i|x_t)\right]_i,$
where $x_{t-\Delta t}^i$ is the token sampled at position $i$ from $p_\theta(\cdot\mid x_t)$ at step $t$. Please refer to \citet{nie_llada} for training details.%
\footnote{Note that \citet{nie_llada} also have a temperature parameter which we absorb into $p_\theta$.}

\section{Watermarking Scheme}

To motivate our scheme, consider the soft watermark for autoregressive LLMs proposed by \citet{kirchenbauer_water}. This method applies a bias $\delta$, favoring a fraction of the vocabulary $\gamma$ (called the ``green list'') during sampling, with the goal of inducing a statistically detectable excess of green-list tokens in the output. Note that this scheme does not directly translate from autoregressive LLMs to DDLMs. Firstly, unlike autoregressive generation, discrete diffusion samples tokens across a set of diffusion steps $S$, requiring the selection of a subset $S_W \subseteq S$ at which to apply watermarking. This introduces a completeness–distortion tradeoff: we observe empirically that as $|S_W|$ increases, completeness and distortion are both non-decreasing.  Secondly, \citet{kirchenbauer_water} seed the green-list construction using previously generated tokens, an approach that is incompatible with diffusion models where all tokens are generated concurrently rather than sequentially.

The nuisance of finding a suitable subset of diffusion steps for watermarking motivates a scheme that is unbiased at every step of the sampling process. In this case, $S_W=S$ is optimal---if accumulating bias that distorts the text is no longer a concern, we should watermark every step to maximize completeness. To do so, we take inspiration from \citet{aaronson2022watermark} that sample the token with the maximum value of $r^{\frac{1}{p}}$ such that $r \sim \text{Unif}[0,1]$ and $p$ is the probability of that token. In Alg. \ref{alg:algorithm}, we repeat this custom sampling at every sampling step. In Appendix~\ref{app:theory}, we prove two key theoretical guarantees: first, that the Gumbel-max sampling rule preserves the original token sampling distribution at every diffusion step, ensuring distortion-freeness (Theorem~\ref{thm:distortion-free}); second, the resulting detection statistic has an exponentially decaying false detection probability in the generated sequence length (Theorem~\ref{thm:prob-false-alarm}). Crucially, the method obviates the need for explicit optimization and fine-tuning of hyperparameters (i.e.\ $\gamma^*,\delta^*,S_W^*$) to achieve an appropriate detectability-distortion balance, which is model and benchmark dependent.

To address the lack of a sequential dependency structure in diffusion models, we seed the random number generator that produces $r_i$ by the position in the sequence. This way, during detection (Alg.~\ref{alg:algorithm-detect}), we can regenerate $r_i$ at each position $i$. We directly follow \citet{aaronson2022watermark} in computing a score $\frac{1}{L}\sum_{i=1}^{L}\ln(\frac{1}{1-r_i})$ where $L$ is the length of the generated sequence. For unwatermarked text, the normalized score is centered around 1 whereas for watermarked text, it is expected to be centered around a value exceeding some threshold that is greater than 1. To thwart prefix deletions that cause all of the seeding to be off by some offset, we follow \citet{kuditipudi_robust} and instead seed by $(s+i) \bmod m$ for some parameter $m \ll L$ and offset $s$. In detection (Alg. \ref{alg:algorithm-detect}), we iterate through all possible offsets and choose the alignment with the maximum normalized score.%
\footnote{This modification bolsters robustness, but note that it does change the distribution of $r$.}

\begin{figure}[t]
\centering

\begin{minipage}{0.48\textwidth}
\begin{algorithm}[H]
\caption{Generate watermarked text}
\label{alg:algorithm}
\textbf{Input:} Diffusion model $p_\theta$, wm parameter $m$ \\
\textbf{Output:} Output sequence $x \in \mathcal{V}^d$
\begin{algorithmic}[1]
\STATE Initialize masked sequence $x_1 \gets MM\ldots(\text{$d$ times})$
\FOR{diffusion steps $t \in \{1, 1-\Delta t, \ldots, 0\}$}
    \STATE Run model $p_{i,x} \gets [p_\theta(\cdot \mid x_t)]_{i,x}$
    \FOR{$i \in \{1,\ldots,d\}$}
        \STATE Draw $r_i \sim \text{Unif}[0,1]^{|\mathcal{V}|}$ with seed $i \bmod m$
        \STATE Compute $S_{i,x} \gets \frac{\ln r_{i,x}}{p_{i,x}}$
        \STATE Set $x_{t-\Delta t,i} \gets \arg\max_x S_{i,x}$
    \ENDFOR
    \STATE Mask the $td$ tokens with least $\{p_{i,x_{t-\Delta t,i}}\}_{i=1}^d$
\ENDFOR
\RETURN $x_0$
\end{algorithmic}
\end{algorithm}
\end{minipage}
\hfill
\begin{minipage}{0.48\textwidth}
\begin{algorithm}[H]
\caption{Detect watermark in text}
\label{alg:algorithm-detect}
\textbf{Input:} Tokens $x_0 \in \mathcal{V}^L$, threshold $\tau$, wm parameter $m$ \\
\textbf{Output:} Watermarked / Not watermarked
\begin{algorithmic}[1]
\FOR{$s \in \{0,1,\ldots,m-1\}$}
    \STATE $\Gamma \gets 0$
    \FOR{$i \in \{1,\ldots,L\}$}
        \STATE Draw $r_i \sim \text{Unif}[0,1]^{|\mathcal{V}|}$ with seed $(i+s) \bmod m$
        \STATE $\Gamma \gets \Gamma - \ln(1 - r_{i,x_{0,i}})$
    \ENDFOR
    \IF{$\Gamma/L > \tau$}
        \RETURN Watermarked
    \ENDIF
\ENDFOR
\RETURN Not watermarked
\end{algorithmic}
\end{algorithm}
\end{minipage}
\label{fig:algorithms}
\end{figure}

\section{Experimental Results}

We prompt LLaDA with open-ended prompts from WaterBench~\citep{tu2024waterbenchholisticevaluationwatermarks}, evaluating text quality using perplexity (via GPT-2~\citep{radford2019gpt2}) and GPT-4-as-a-judge scores (Appendix~\ref{app:prompt}). As a baseline, we adapt the green-list watermark to DDLMs (Appendix~\ref{app:rg-bias}), conducting an ablation over $\gamma \in \{0.1, 0.25, 0.5, 0.75, 0.9\}$, $\delta \in \{0.5, 1, 2, 4, 8\}$, and watermarking intervals $S_W = \{t_{\mathrm{start}} = 1,, t_{\mathrm{end}}\}$ with $t_{\mathrm{end}} \in \{5, 10, 20, 40, 80, 160, 300\}$ (Appendix~\ref{app:rg-results}). Consistent with \citet{kirchenbauer_water}, we consider a watermark to be detectable if the green-list $z$-score is greater than 4. To identify the hyperparameters that balance completeness and distortion-freeness, for a held-out set of prompts, we select the $(\gamma, \delta, S_W)$ that maximizes the GPT-4-as-a-judge score such that $\mathbb{P}(z \ge 4) \ge 1 - \beta$, where $\beta \in \{0.01, 0.05, 0.10, 0.15\}$. By contrast, the Gumbel-max approach we introduce in Alg.\ \ref{alg:algorithm} is far less computationally expensive, as the only hyperparameter to tune is the detection threshold $\tau$.\footnote{Tuning $S_W$ is unnecessary for the Gumbel-max approach as the watermark is unbiased at each step (Theorem \ref{thm:distortion-free}). We show the selected $S_W$ has no impact on judge scores in Fig.\ \ref{fig:aaronson-by-step}.} A $\tau$ that is too large decreases completeness (high Type II error $\beta$) whereas a $\tau$ that is too small decreases soundness (high Type I error $\alpha$). Using a held-out set of prompts, we choose $\tau^*=1.19$ as it maximizes $1-\beta$ such that $\alpha \leq 0.01$ (Fig. \ref{fig:normalized_score_threshold_analysis}).

\begin{table}[t]
\centering
\small
\begin{tabular}{l l cc}
\hline
\textbf{Method} & \textbf{Hyperparameters} & \textbf{Completeness} & \textbf{Soundness} \\
\hline

Green-List ($\beta=0.15$)
& $\gamma^*=0.25,\ \delta^*=4.0,\ t_{start}=1,\ t_{end}^*=80$
& 80.6\% 
& 97.2\% \\

Green-List ($\beta=0.10$)
& $\gamma^*=0.25,\ \delta^*=8.0,\ t_{start}=1,\ t_{end}^*=20$ 
& 94.0\% 
& 97.2\% \\

Green-List ($\beta=0.05$)
& $\gamma^*=0.1,\ \delta^*=8.0,\ t_{start}=1,\ t_{end}^*=20$  
& 98.3\% 
& 97.2\% \\

Green-List ($\beta=0.01$)
& $\gamma^*=0.1,\ \delta^*=8.0,\ t_{start}=1,\ t_{end}^*=20$ 
& 98.3\%  
& 97.2\% \\

Gumbel-max 
& $\tau^*=1.19$ 
& \textbf{96.0\%}
& \textbf{97.7\%} \\
\hline
\end{tabular}
\caption{Detectability performance of watermarking methods on 1000 WaterBench prompts. The hyperparameter selection method from the held-out set is specified in parenthesis. The Gumbel-max method attains comparable completeness and soundness compared to green-list baselines which necessitate extensive hyperparameter tuning.}
\label{tab:watermark_detectability}
\end{table}

\begin{table}[t]
\centering
\small
\setlength{\tabcolsep}{4pt}
\begin{tabular}{l ccccccc}
\hline
\textbf{Method} 
& \textbf{PPL} 
& \textbf{Style} 
& \textbf{Consistency} 
& \textbf{Accuracy} 
& \textbf{Ethics} 
& \textbf{Avg. Score} \\
\hline
No Watermark 
& 10.391 
& 8.508 
& 8.954 
& 7.364 
& 8.474 
& 8.325 \\

Green-List ($\beta=0.15$)
& 12.195 
& 8.457 
& 8.782 
& 7.044 
& 8.320 
& 8.151 \\

Green-List ($\beta=0.10$)
& 10.485 
& 8.302 
& 8.574 
& 6.978 
& 8.290 
& 8.036 \\

Green-List ($\beta=0.05$)
& 9.878 
& 7.399 
& 7.484 
& 6.214 
& 7.782 
& 7.220 \\

Green-List ($\beta=0.01$)
& 9.878 
& 7.399 
& 7.484 
& 6.214 
& 7.782 
& 7.220 \\

Gumbel-max 
& 10.652 
& 7.922 
& 7.849 
& 6.712 
& 7.904 
& 7.597 \\
\hline
\end{tabular}
\caption{Generation quality evaluated for the same 1000 WaterBench prompts and hyperparameters as in Table~\ref{tab:watermark_detectability}. Green-list methods exhibit increasing quality degradation as detectability improves. The Gumbel-max approach incurs only a modest reduction in generation quality.}
\label{tab:watermark_quality}
\end{table}

\begin{figure}[t]
\centering

\begin{minipage}{0.48\linewidth}
  \centering
  \includegraphics[width=\linewidth]{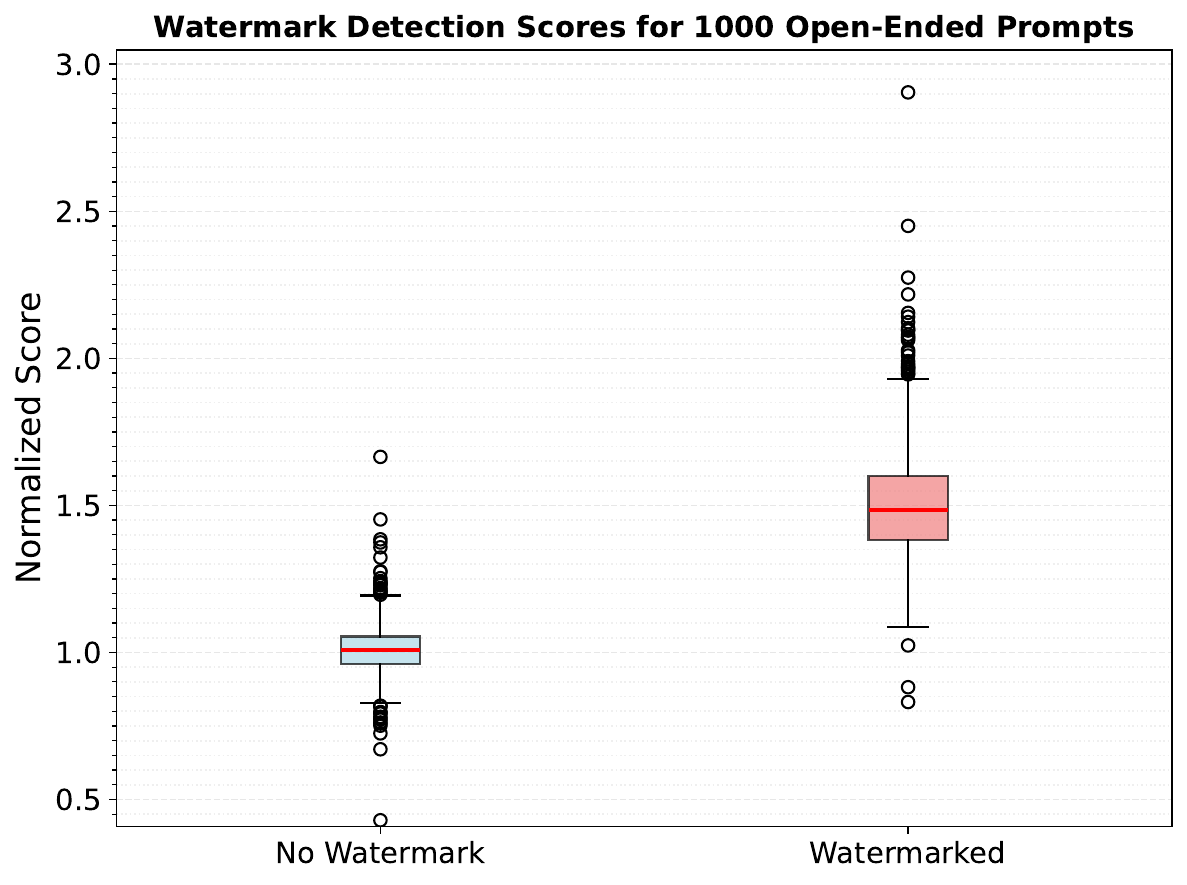}
  \caption{Distribution of normalized detection scores for unwatermarked and watermarked text using the Gumbel-max scheme. Results are shown for 1000 WaterBench prompts, illustrating a clear statistical difference.}
  \label{fig:watermark_boxplot}
\end{minipage}
\hfill
\begin{minipage}{0.48\linewidth}
  \centering
  \includegraphics[width=\linewidth]{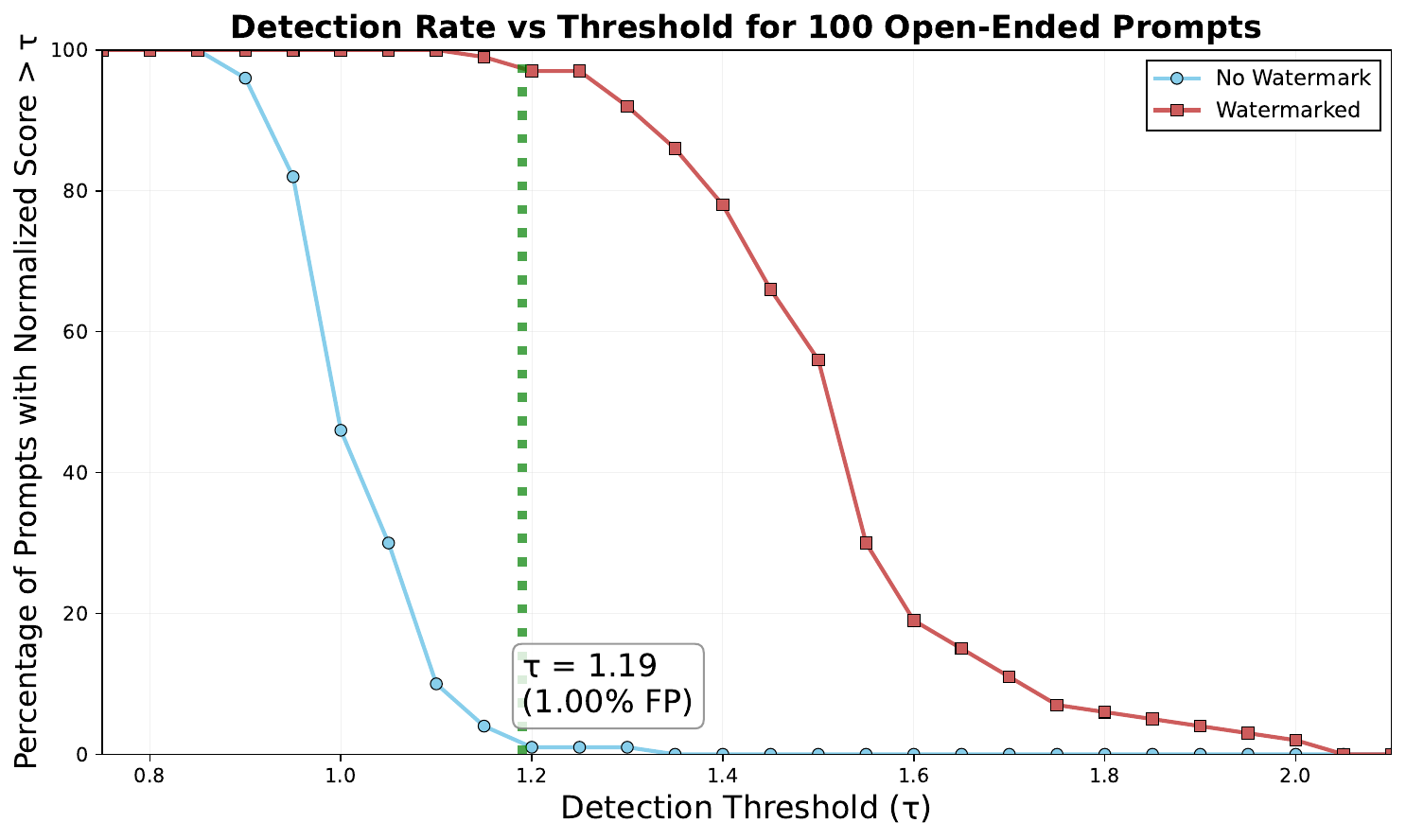}
  \caption{Percentage of open-ended prompts exceeding threshold $\tau$ for varying $\tau$, shown for unwatermarked and watermarked text. Results are computed on a held-out set of 100 WaterBench prompts, illustrating the soundness--completeness trade-off.}
  \label{fig:normalized_score_threshold_analysis}
\end{minipage}

\end{figure}

The green-list watermark is detectable but it exhibits sensitivity to the hyperparameter selection strategy. Table~\ref{tab:watermark_detectability} and Table~\ref{tab:watermark_quality} indicate that increasing completeness (i.e.\ decreasing $\beta$) for green-list watermarking generally degrades text quality, with the $\beta=0.10$ configuration providing the most favorable distortion--detectability trade-off. The Gumbel-max method also achieves strong detectability, with Fig.~\ref{fig:watermark_boxplot} highlighting a clear statistical separation between the normalized detection score distributions of watermarked and unwatermarked text. Specifically, the Gumbel-max watermark achieves 96.0\% completeness and 97.7\% soundness (Table~\ref{tab:watermark_detectability}) with $\tau^*$. Moreover, it maintains detectability under prefix deletions (Appendix~\ref{app:robustness}). Compared to the $\beta=0.10$ green-list baseline, it offers stronger detectability with only a modest reduction in text quality (Table~\ref{tab:watermark_quality}) while requiring substantially simpler tuning. 

\section{Discussion and Conclusion}

We introduce a DDLM watermarking scheme and demonstrate its completeness empirically, as well as its soundness and distortion-freeness both theoretically and empirically. Our method is also robust to prefix deletions, maintaining strong detectability. Unlike green-list-based approaches, which require expensive ablations to balance detectability and distortion, our method avoids hyperparameter search and thus enables straightforward deployment across models and benchmarks. One promising direction to improve the watermarked text quality using our method is more deeply considering the re-masking strategy and concentrating watermark signal in regions that are more stable under the reverse diffusion dynamics. There is also scope to explore edit-distance-based alignment techniques~\citep{kuditipudi_robust} that provide robustness guarantees beyond prefix deletions, while preserving distortion-freeness guarantees of our method. Future work must also evaluate our method against contemporaneous  approaches such as \citet{gloaguen2025watermarkingdiffusionlanguagemodels}. 

\bibliography{iclr2026_conference}
\bibliographystyle{iclr2026_conference}

\appendix

\section{Appendix}

\subsection{Related Work}\label{app:related-work}

\textbf{Discrete Diffusion Models:} 
Diffusion models gradually corrupt data into a simple distribution and then learn a reverse process to reconstruct the original data \citep{sohl-dickstein15_thermo, ho_denoising}. Although earlier work focused on learning continuous data, diffusion model methods have been adapted for categorical domains. Unlike their autoregressive counterparts that capture left-to-right dependencies, discrete diffusion models generate tokens in parallel, enabling faster inference \citep{wang_faster}, greater controllability \citep{schiff_simple}, and enhanced comprehension of global patterns \citep{hu_global}. While the first discrete diffusion models mapped discrete tokens into continuous space before applying traditional continuous diffusion \citep{li_dlm, gong_difseq, han_ssd}, more recent models operate directly on the discrete state space \citep{hoogeboom_argmax, austin_struct, meng_score, lou_sedd, nie_llada}. Inception Labs’ models empirically validate discrete diffusion at scale across text and multimodal domains \citep{labs2025mercuryultrafastlanguagemodels, li2025lavidalargediffusionlanguage}.

\textbf{Watermarking Generative Models:} A watermark should satisfy the following properties:
\begin{itemize}
    \item \textbf{Soundness}: The detection scheme reliably identifies unwatermarked content as unwatermarked \citep{CiC-1-4-31, christ_water, fan_two}.
    \item \textbf{Completeness}: The detection scheme reliably identifies watermarked content as watermarked \citep{CiC-1-4-31, christ_water}.
    \item \textbf{Distortion-Freeness}: The watermark does not significantly reduce the quality of the text \citep{CiC-1-4-31, christ_water, giboulot_tradeoff}. Completeness and distortion-freeness are often inversely related. 
    \item \textbf{Robustness}: The watermark is still detectable following bounded modifications to the text \citep{CiC-1-4-31, giboulot_tradeoff, li_robust, kuditipudi_robust, fan_two}.
\end{itemize}
In continuous settings (images, audio, etc.), watermarking is well-studied as one can exploit known signal processing properties to achieve a robust fingerprint \citep{wen_tree, katzenbeisser_book}. Accordingly, existing watermarking approaches for diffusion models \citep{wen_tree, liu2023watermarkingdiffusionmodel, liu2024grootgeneratingrobustwatermark, zhao2023recipewatermarkingdiffusionmodels, huang2025robinrobustinvisiblewatermarks} focus almost exclusively on continuous modalities. Naive modifications to adapt such watermarks to discrete settings, however, yield incomprehensible and undetectable outputs. \citet{kirchenbauer_water} introduce a watermark for models generating discrete data (i.e., autoregressive LLMs) which biases the sampling procedure to favor a subset of the vocabulary (i.e., the ``green list'') with the hope that there is a statistically greater number 
of tokens from the green list in the output. More recently, \citet{weiji_su} introduce a hypothesis-testing framework to systematically evaluate and compare different watermark detection methods. \citet{he_theory} proposed jointly optimizing both watermark embedding and detection, detailing the trade-off between detectability and distortion---a balance we also empirically reveal in this paper. Contemporaneously to the present work, \citet{gloaguen2025watermarkingdiffusionlanguagemodels} address this tradeoff by formulating DDLM watermarking as a constrained optimization problem, applying green-list watermarking in expectation over unknown context during generation.

\subsection{Theoretical Results}\label{app:theory}

In this section, we first analyze the scheme in Alg.~\ref{alg:algorithm} and prove its distortion-free property
based on the results from \citet{aaronson2022watermark} in Theorem~\ref{thm:distortion-free}.
We then show that the probability of false detection by Alg.~\ref{alg:algorithm-detect} 
exponentially goes to $0$ as a function of the number of tokens
in the sequence in Theorem~\ref{thm:prob-false-alarm}.

\begin{theorem}
    Given a diffusion language model $p_\theta$, the output text of 
    Alg.~\ref{alg:algorithm} has the same distribution as that of the
    unwatermarked language model $p_\theta$ if the effects of the pseudorandom seed are negligible.
    \label{thm:distortion-free}
\end{theorem}
\begin{proof}
    At each time $t\in[0,1]$, the (unwatermarked) model samples the output according to $p_\theta(\cdot\mid x_t)$ (see Section \ref{sec:llada})
    before masking them and proceeding to the next step $t-\Delta t$.
    In our watermarking scheme, this is replaced with Steps 5-7
    of Alg.~\ref{alg:algorithm}.
    If we assume the effects of the pseudorandom seed repeating after every $m$ positions is negligible,
    we can follow the Gumbel-max derivation of \citet{aaronson2022watermark} to show that the output
    has the same distribution as the equation in Section~\ref{sec:llada} even in Alg.~\ref{alg:algorithm}.
    We have a set of $|\mathcal{V}|$ i.i.d. $\text{Unif}[0,1]$ variables and we need to show that
    picking $\arg\max_y\frac{\ln R_y}{p_y}$ has the same distribution as $p_y$ for $y\in\{1,2,\ldots,|\mathcal{V}|\}$.
    Writing out the expressions, we get
    \begin{align*}
        \mathbb{P}(Y=y) &= \mathbb{P}\left(\frac{\ln{R_y}}{p_y} \ge \frac{\ln R_z}{p_z}\ \forall\ z\neq y\right) \\
        &= \int_0^1 \prod_{z\neq y}r_y^{p_z/p_y}dr_y \tag{$\{R_z\}$ independent} \\
        &= \int_0^1 r_y^{\frac{1-p_y}{p_y}}dr_y = p_y,
    \end{align*}
    which concludes the proof.
\end{proof}

\begin{theorem}
    Alg.~\ref{alg:algorithm-detect}  with parameter $\tau=1+\zeta$ applied on unwatermarked text
    satisfies
    \begin{align*}
        \mathbb{P}&\left(\text{Watermark detection}\mid \text{Unwatermarked text}\right) \le m\exp\left(-L(\zeta-\ln(1+\zeta))\right),
    \end{align*}
    where $L$ is the number of tokens in the sequence.
    \label{thm:prob-false-alarm}
\end{theorem}
\begin{proof}
    We use a familiar Chernoff bound to prove this.
    When the tokens are unwatermarked, the variables $r_{i,x}$ are all i.i.d.\
    $\text{Unif}[0,1]$ random variables since we are using different seeds for every $i$,
    for a given $s$.
    This gives us for some i.i.d.\ $R_i\sim\text{Unif}[0,1]$:
    \begin{align*}
        \mathbb{P}&\left(\text{Detect}\mid \text{No wm}, s\right) 
        = \mathbb{P}\left(-\frac{1}{L}\sum_{i=1}^L\ln(1-R_i) > 1+\zeta\right) \\
        &= \mathbb{P}\left(\exp\left(\lambda\sum_{i=1}^L-\ln(1-R_i)\right)>\exp(L\lambda(1+\zeta))\right) \tag{for $\lambda>0$}\\
        &\le \exp\left(-L\lambda(1+\zeta)\right)\mathbb{E}\left[\prod_{i=1}^L\frac{1}{(1-R_i)^\lambda}\right],
    \end{align*}
    where the final expression follows from Markov inequality.
    Using the independence of $\{R_i\}$, this reduces to
    \begin{align*}
        \mathbb{P}&\left(\text{Detect}\mid \text{No wm}, s\right) \le
        \exp\left(-L\lambda(1+\zeta)\right)\left(\int_0^1\!\!\!\frac{1}{(1-r)^\lambda}dr\right)^L.
    \end{align*}
    Since the above is true for any $\lambda>0$, we can choose the $\lambda$ that minimizes the
    expression. For $0<\lambda<1$, the integral simplifies to
    \begin{align*}
        \mathbb{P}(\text{Detect}\mid\text{No wm}, s) \le \frac{1}{(1-\lambda)^L}\exp\left(-L\lambda(1+\zeta)\right).
    \end{align*}
    Minimizing this gives us $\lambda=\frac{\zeta}{1+\zeta}$, which results in
    \begin{align*}
        \mathbb{P}(\text{Detect}\mid\text{No wm}, s) \le \exp\left(-L(\zeta-\ln(1+\zeta))\right).
    \end{align*}
    Applying a union bound over all $s\in\{0,1,\ldots,m-1\}$ concludes the proof.
\end{proof}

\subsection{Robustness}\label{app:robustness}

\begin{table}[H]
\centering
\small
\begin{tabular}{c c c}
\hline
\textbf{Parameter $m$} & \textbf{Avg. Judge score} & \textbf{Completeness (\%)} \\
\hline
No Watermark & 7.0910 & 2.2\% \\
2 & 5.9595 & 98.4\% \\
3 & 6.3340 & 96.4\% \\
5 & 6.4815 & 98.4\% \\
7 & 6.7380 & 98.6\% \\
10 & 6.7850 & 98.4\% \\
\hline
\end{tabular}
\caption{Generation quality (average GPT4-as-a-judge score) and watermark completeness (percentage of samples with $\tau \geq 1.19$) across robustness parameter values $m$ (Alg.~\ref{alg:algorithm-detect}) on 500 WaterBench prompts. Generated tokens are subjected to random prefix deletions. Completeness remains high for all $m \geq 2$. As $m$ increases, seed reuse becomes less frequent causing the joint distribution of generated tokens to more closely match that of the unwatermarked model, explaining the monotonic increase in quality metrics.}
\label{tab:judge_completeness_vs_m}
\end{table}

\subsection{Green-List Watermark Method}\label{app:rg-bias}

In the main exposition, we presented an unbiased scheme to watermark the sampling procedure of discrete diffusion models. In this section, we instead directly bias the sampling procedure to favor a subset of the vocabulary. To do so, we follow \citet{kirchenbauer_water} by partitioning the vocabulary into a green and red list of sizes $\gamma|V|$ and $|V|-\gamma|V|$ respectively. We seed the partition by the position in the sequence $i$ to maintain a diversity in the tokens being sampled (i.e. the green list changes at each index in the sequence) and so that a receiver aware of the scheme can reconstruct the green list at each position. In Algorithm \ref{alg:rg-mask}, we generate a green mask of size $\gamma|V|$ uniformly at random, again following \citet{kuditipudi_robust} by using a modulus parameter to thwart prefix deletions.  

\begin{algorithm}[h]
\caption{Generate green mask}
\label{alg:rg-mask}
\textbf{Input:} Sequence length $L$, vocabulary size $|\mathcal{V}|$, green list size $\gamma$, modulus $m$, starting offset $s$ \\
\textbf{Output:} Mask $G \in \{0,1\}^{L \times |\mathcal{V}|}$
\begin{algorithmic}[1]
\STATE Initialize empty matrix $G \in \{0,1\}^{L \times |\mathcal{V}|}$
\FOR{$i \in \{1...L\}$}
    \STATE Choose each $G[i][j]$ uniformly at random with seed $(s+i) \bmod m$ such that $\sum_{j=1}^{|V|}G[i][j]=\gamma|V|$
\ENDFOR
\RETURN $G$
\end{algorithmic}
\end{algorithm}

Taking inspiration from the soft watermark in \citet{kirchenbauer_water}, we add a bias to favor the green logits during sampling (Algorithm \ref{alg:rg-generate}, Figure \ref{fig:llada-diagram}). Where $l=\operatorname{logit} (p_{t-\Delta t \mid t}(x_{t-\Delta t}^i \mid x_t^i))$,
\begin{align}
p'_{t-\Delta t \mid t}(x_{t-\Delta t}^i \mid x_t^i) =
\operatorname{Softmax}\!\bigl(l + \delta \cdot \mathbf{1}_{\{x_{t-\Delta t}^i \in G\}}\bigr).
\end{align}

\begin{algorithm}[tb]
\caption{Generate watermarked text}
\label{alg:rg-generate}
\textbf{Input:} Diffusion  model $p_\theta$, wm  parameter $m$ \\
\textbf{Output:} Output sequence $x\in \mathcal{V}^d$
\begin{algorithmic}[1]
\STATE Initialize masked sequence $x_1\gets MM\ldots(\text{$d$ times})$
\STATE $G=\operatorname{GenerateGreenMask}(L,|V|, \gamma,m,0)$
\FOR{diffusion steps $t \in \{1, 1-\Delta t, 1-2\Delta t\ldots0\}$}
        \STATE Run model $p_{i,x}\gets \left[p_\theta(\cdot\mid x_t)\right]_{i,x}\ \forall i\in\{1,\ldots,d\}$
        \STATE $l=\operatorname{logit} (p_{t-\Delta t \mid t}(x_{t-\Delta t}^i \mid x_t^i))$
        \STATE $x_{t-\Delta t,i} \sim \operatorname{Softmax}\!\bigl(l + \delta \cdot G[i][x_{t-\Delta t}^i])$
        \STATE Mask the $td$ tokens with least $\{p_{i,x_{t-\Delta t,i}}\}_{i=1}^d$ 
\ENDFOR
\RETURN $x_0$
\end{algorithmic}
\end{algorithm}

During detection in Algorithm \ref{alg:rg-detect}, we regenerate the green mask at every position, calculating whether the proportion of green tokens is statistically significant (i.e. z-score).  

\begin{algorithm}[tb]
\caption{Detect watermark in text}
\label{alg:rg-detect}
\textbf{Input:} Tokens $x_0 \in \mathcal{V}^L$, threshold $\tau=4$, wm parameter $m$\\
\textbf{Output:} Watermarked / Not wm
\begin{algorithmic}[1]
\FOR{$s \in \{0, 1, \ldots, m-1\}$}
    \STATE $\mathrm{count} \leftarrow 0$
    \STATE $G=\operatorname{GenerateGreenMask}(L,|V|,\gamma,m,s)$
    \FOR{$i\in \{1,2,\ldots, L\}$}
        \IF{$G[i][x_{0,i}]==1$}
            \STATE $\mathrm{count}++$
        \ENDIF
    \ENDFOR
    \STATE $z=\frac{\mathrm{count}-\gamma L}{\sqrt{\gamma L(1-\gamma)}}$
    \IF{$z > \tau$}
        \RETURN Watermarked
    \ENDIF
\ENDFOR
\RETURN Not Watermarked
\end{algorithmic}
\end{algorithm}

Note that unlike the unbiased approach, our choice of hyperparameters now affects the tradeoff between detectability and distortion-freeness. Specifically, a high $\delta, |S_W|$ (where $S_W \subseteq S$ is the set of watermarked steps) increases detectability but further distorts the text. Thus, we must perform a sweep over possible $\gamma$, $\delta$, and $S_W$. Let $D$ be the resultant dataset and $A$ the event that the watermark detector Algorithm \ref{alg:rg-detect} returns true. The objective is as follows where $G$ is the number of recovered green tokens. Soundness and completeness lower bounds are chosen at $95\%$ and $z\geq4$ respectively to match prior work.

\begin{equation}
\begin{aligned}
\min_{\gamma, \delta, S_W} \quad & \text{Perplexity} \\
\text{s.t.} \quad 
& z = \frac{G - \gamma T}{\sqrt{T \gamma (1 - \gamma)}} \ge 4, \quad
0 < \gamma < 1, \quad \delta > 0, \\
& S_W \subseteq S=[1,\dots,T], \quad
S_W = \{t_{\text{start}}, \dots,t_{\text{end}}\}, \\
& 1 \le t_{\text{start}} \le t_{\text{end}} \le T, \quad
\mathbb{P}_{x \sim D_{\delta=1}}[A] > 0.95
\end{aligned}
\end{equation}

\begin{figure}[t]
\centering
\fbox{%
\includegraphics[width=0.95\columnwidth]{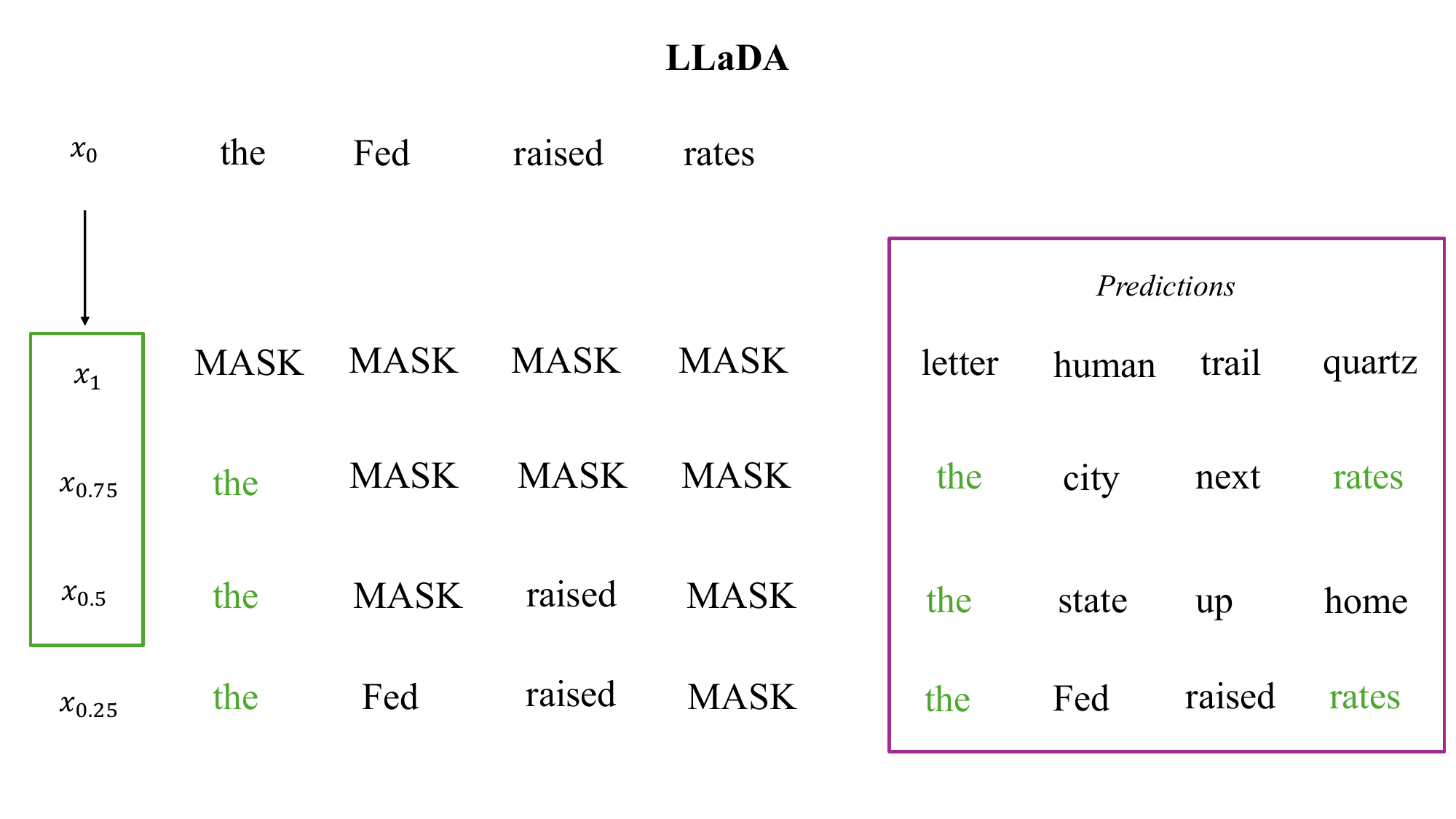}} 
\caption{The forward process converges to an all masked state. LLaDA predicts the entire sequence at each step and then re-masks the $t|V|$ tokens with lowest confidence. Only $M \to v \neq M$ transitions occur in sampling.}
\label{fig:llada-diagram}
\end{figure}

\subsection{Green-List Watermark Results}\label{app:rg-results}

This section reports an ablation over the following parameters for LLaDA:
\[
\gamma \in \{0.1, 0.25, 0.5, 0.75, 0.9\}, \quad
\delta \in \{0.5, 1, 2, 4, 8\}, \quad
t_{\mathrm{end}} \in \{5, 10, 20, 40, 80, 160, 300\},
\]
with step intervals $S_W = \{t_{\mathrm{start}} = 1,\, t_{\mathrm{end}}\}$.
Figures~\ref{fig:fixdelta05_tradeoff}--\ref{fig:fixdelta8_tradeoff} fix $\delta$ and vary
$t_{\mathrm{end}}$, illustrating that watermarking more steps (larger $t_{\mathrm{end}}$)
generally increases detectability ($z$-score) but can degrade generation quality (judge score),
with the trade-off most pronounced for larger $\gamma$.
Conversely, Figures~\ref{fig:fixsteps5_tradeoff}--\ref{fig:fixsteps300_tradeoff} fix
$t_{\mathrm{end}}$ and vary $\delta$, showing that increasing $\delta$ yields stronger detection
but typically at the cost of reduced judge score (and higher distortion).
Overall, higher $\delta$ and larger watermarking ranges tend to improve completeness, while the
associated distortion diminishes as $t_{\mathrm{end}}$ approaches the final diffusion step and the
sequence approaches a fixed state.

\begin{figure*}[t]
  \centering
  \begin{subfigure}{0.48\textwidth}
    \includegraphics[width=\linewidth]{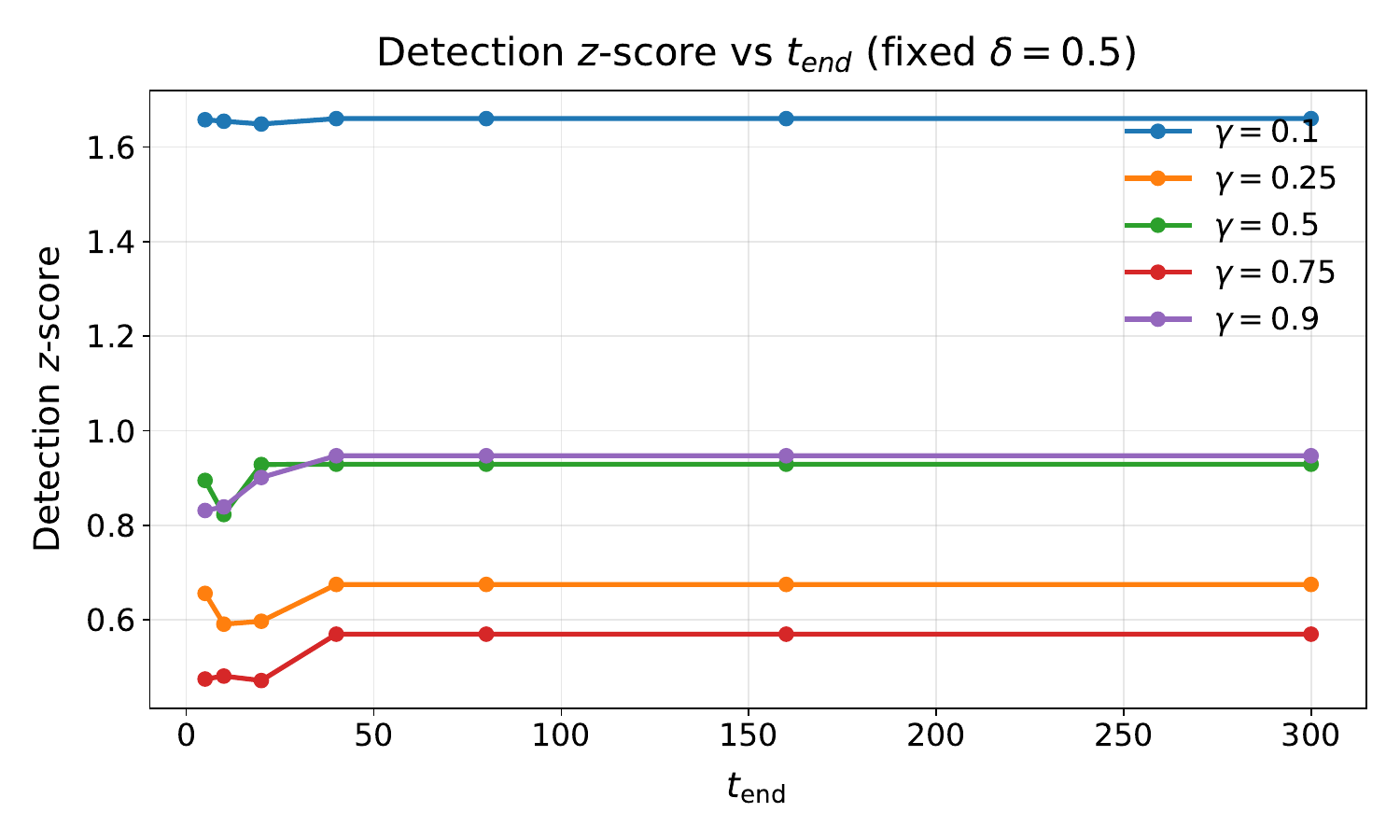}
    \caption*{Detection $z$-score}
  \end{subfigure}\hfill
  \begin{subfigure}{0.48\textwidth}
    \includegraphics[width=\linewidth]{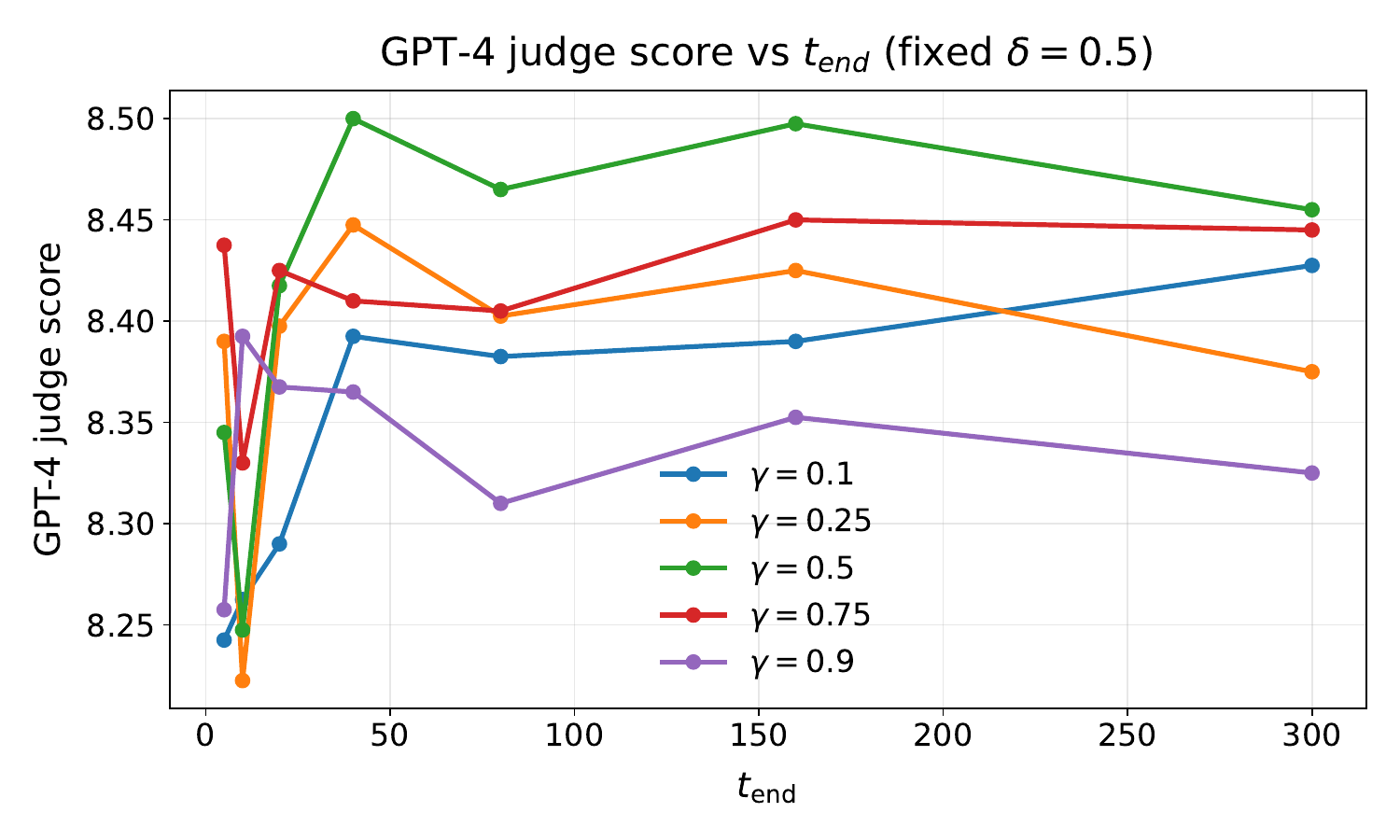}
    \caption*{GPT-4 judge score}
  \end{subfigure}
  \caption{(Green-List watermark) Detection strength vs.\ generation quality for fixed $\delta=0.5$. Performance is shown as a function of $t_{\mathrm{end}}$, with lines corresponding to different $\gamma$.}
  \label{fig:fixdelta05_tradeoff}
\end{figure*}

\begin{figure*}[t]
  \centering
  \begin{subfigure}{0.48\textwidth}
    \includegraphics[width=\linewidth]{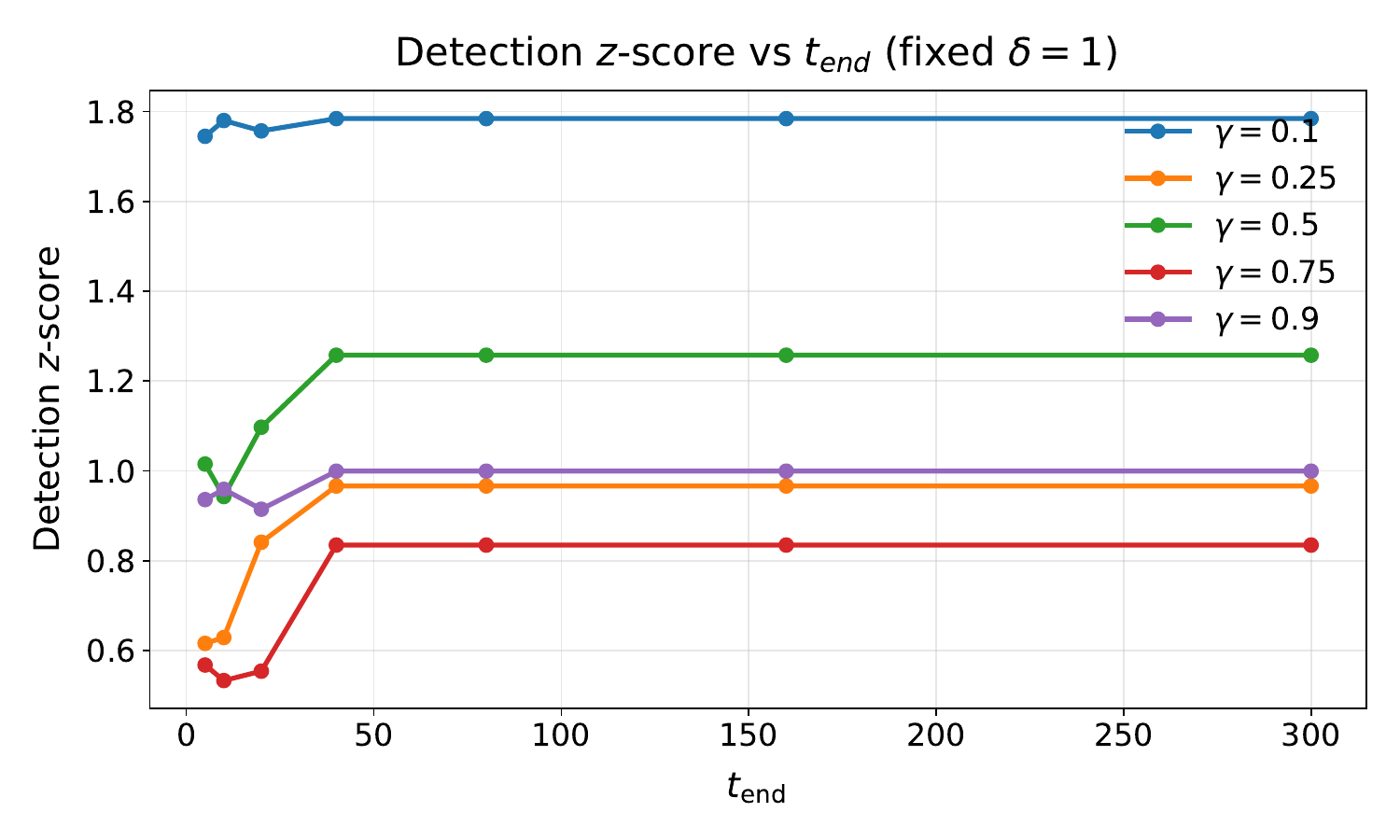}
    \caption*{Detection $z$-score}
  \end{subfigure}\hfill
  \begin{subfigure}{0.48\textwidth}
    \includegraphics[width=\linewidth]{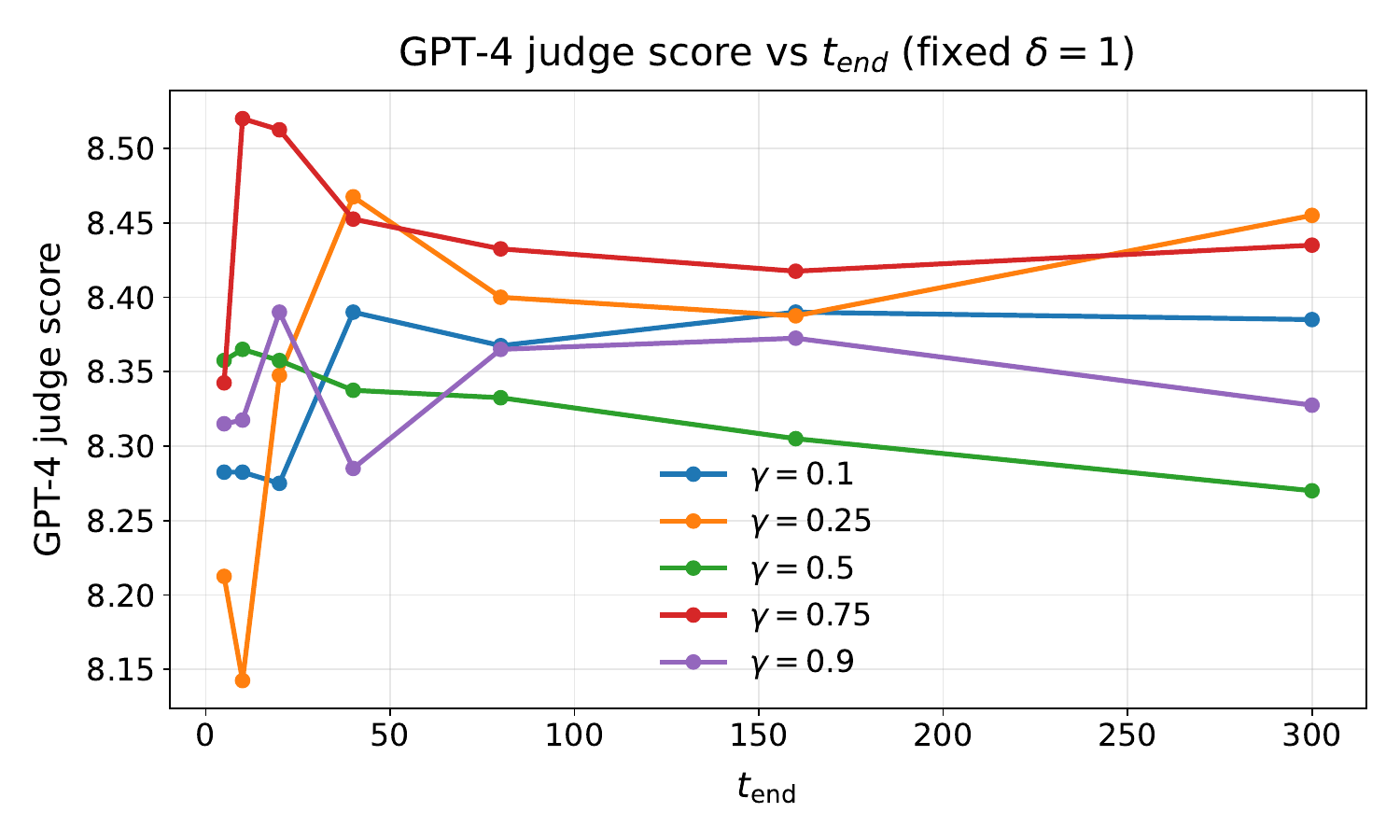}
    \caption*{GPT-4 judge score}
  \end{subfigure}
  \caption{(Green-List watermark) Detection strength vs.\ generation quality for fixed $\delta=1$. Performance is shown as a function of $t_{\mathrm{end}}$, with lines corresponding to different $\gamma$.}
  \label{fig:fixdelta1_tradeoff}
\end{figure*}

\begin{figure*}[t]
  \centering
  \begin{subfigure}{0.48\textwidth}
    \includegraphics[width=\linewidth]{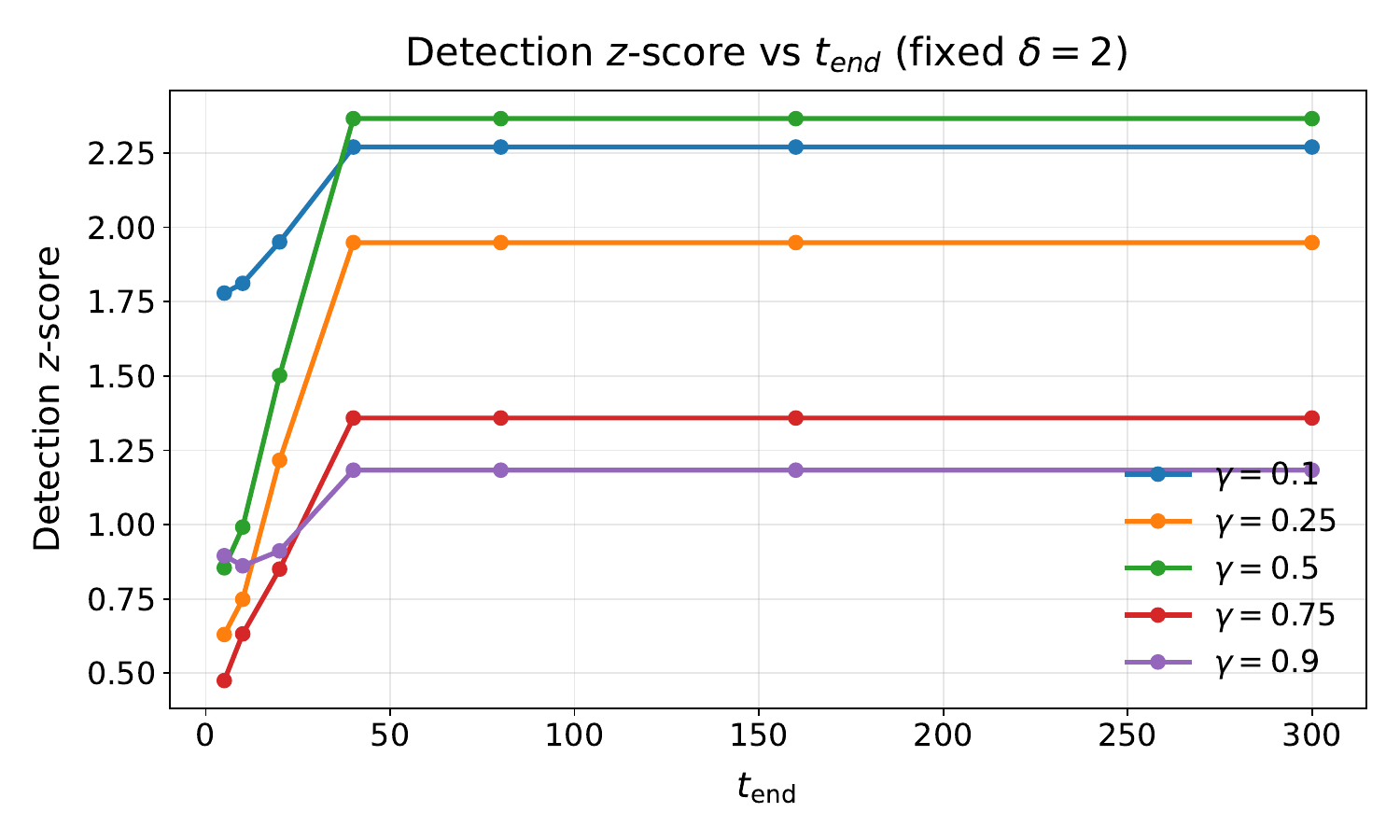}
    \caption*{Detection $z$-score}
  \end{subfigure}\hfill
  \begin{subfigure}{0.48\textwidth}
    \includegraphics[width=\linewidth]{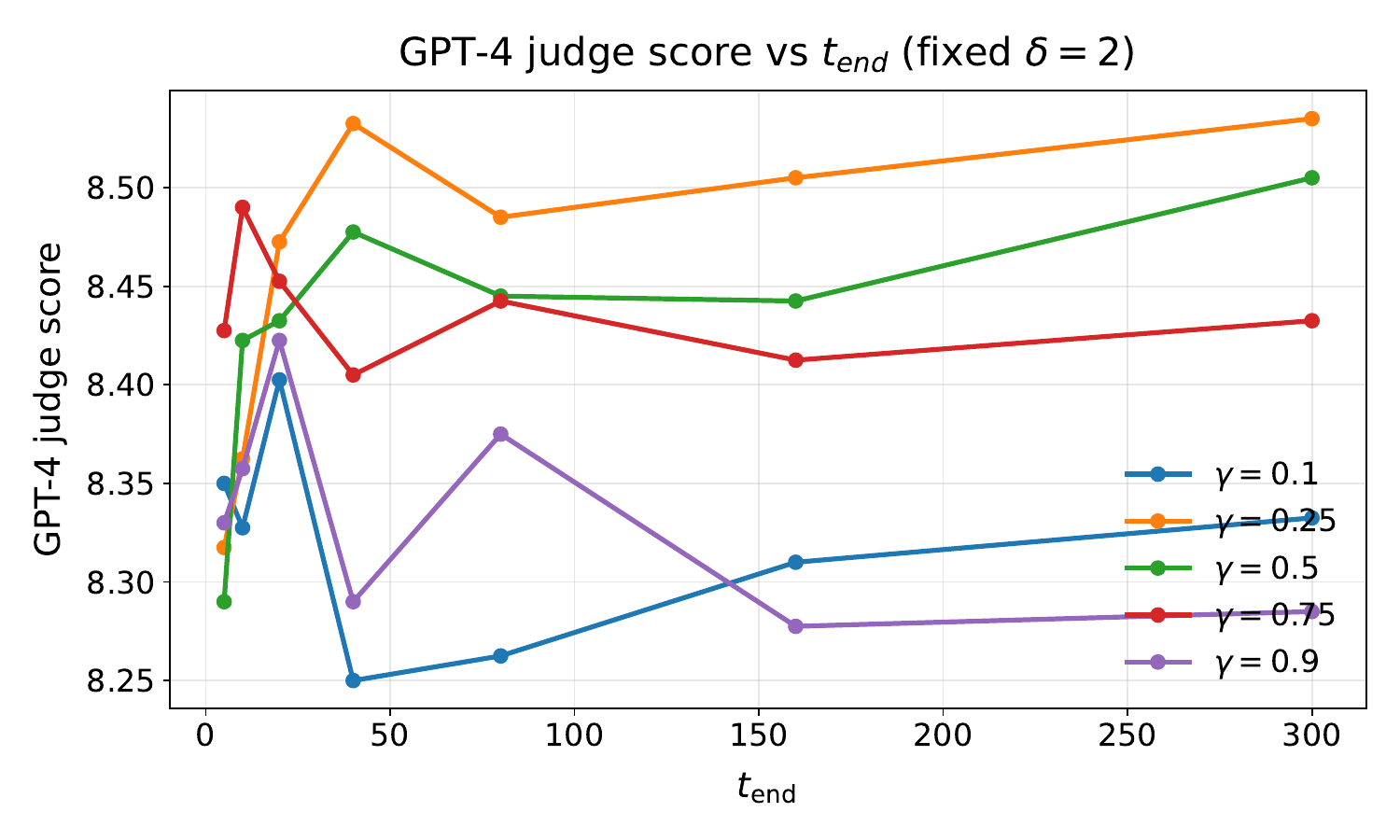}
    \caption*{GPT-4 judge score}
  \end{subfigure}
  \caption{(Green-List watermark) Detection strength vs.\ generation quality for fixed $\delta=2$. Performance is shown as a function of $t_{\mathrm{end}}$, with lines corresponding to different $\gamma$.}
  \label{fig:fixdelta2_tradeoff}
\end{figure*}

\begin{figure*}[t]
  \centering
  \begin{subfigure}{0.48\textwidth}
    \includegraphics[width=\linewidth]{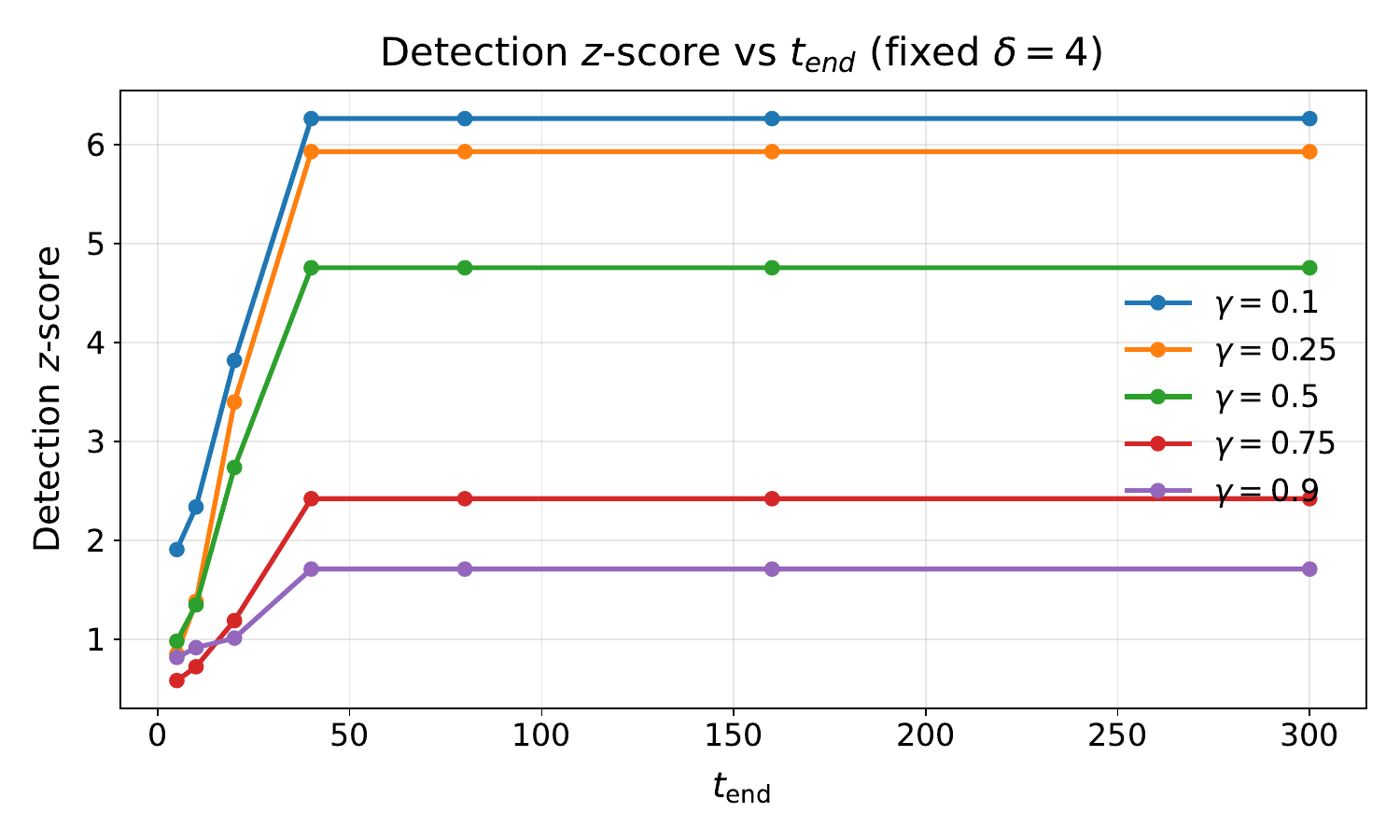}
    \caption*{Detection $z$-score}
  \end{subfigure}\hfill
  \begin{subfigure}{0.48\textwidth}
    \includegraphics[width=\linewidth]{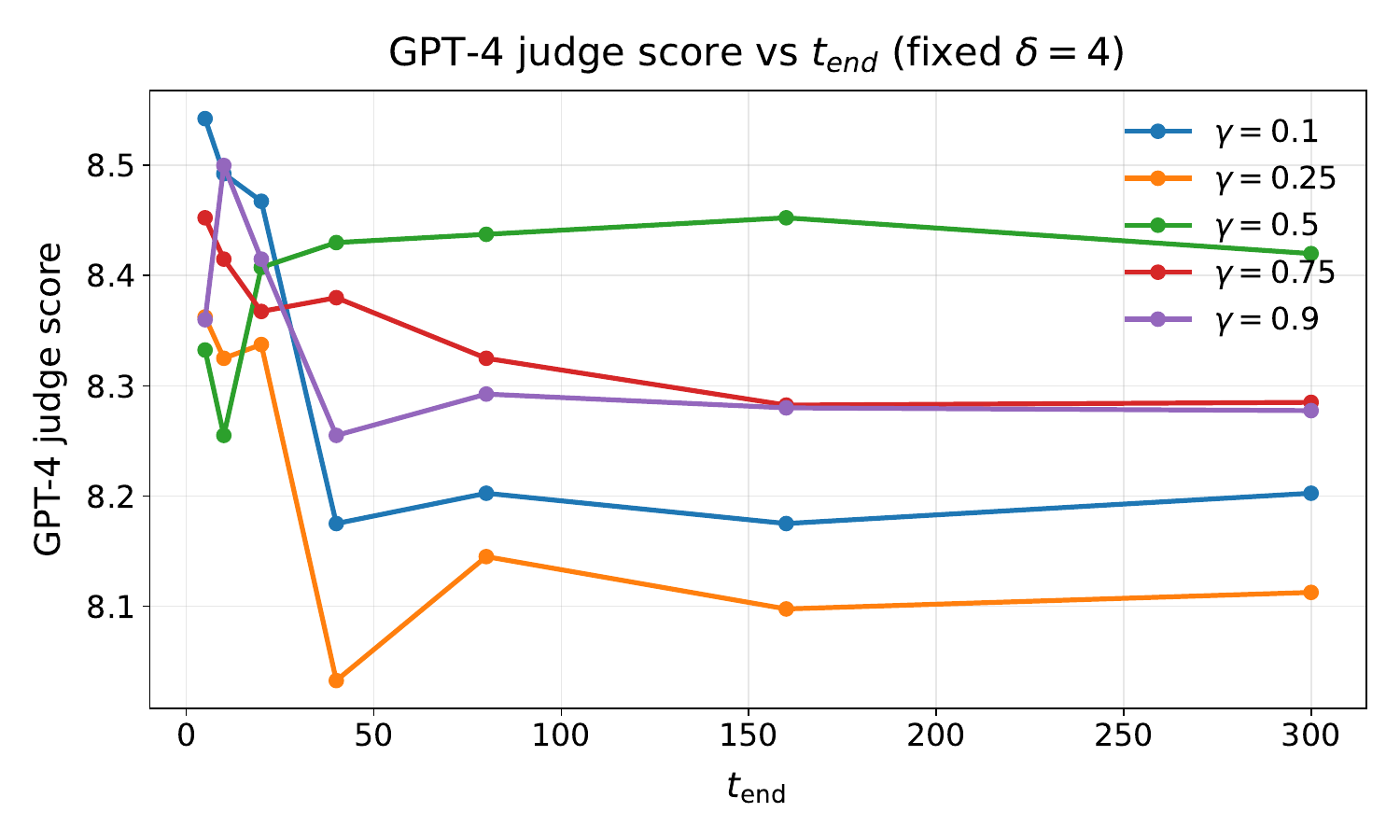}
    \caption*{GPT-4 judge score}
  \end{subfigure}
  \caption{(Green-List watermark) Detection strength vs.\ generation quality for fixed $\delta=4$. Performance is shown as a function of $t_{\mathrm{end}}$, with lines corresponding to different $\gamma$.}
  \label{fig:fixdelta4_tradeoff}
\end{figure*}

\begin{figure*}[t]
  \centering
  \begin{subfigure}{0.48\textwidth}
    \includegraphics[width=\linewidth]{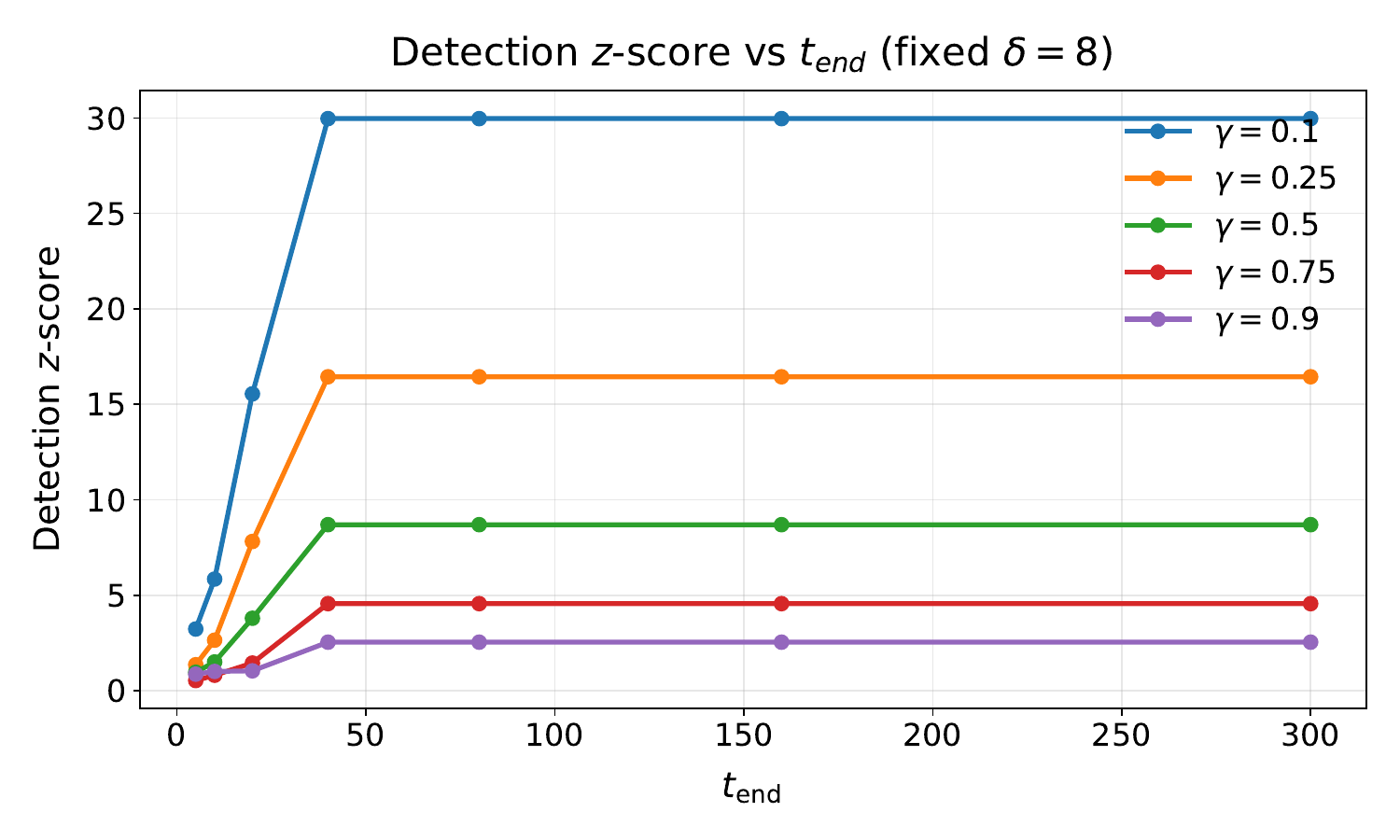}
    \caption*{Detection $z$-score}
  \end{subfigure}\hfill
  \begin{subfigure}{0.48\textwidth}
    \includegraphics[width=\linewidth]{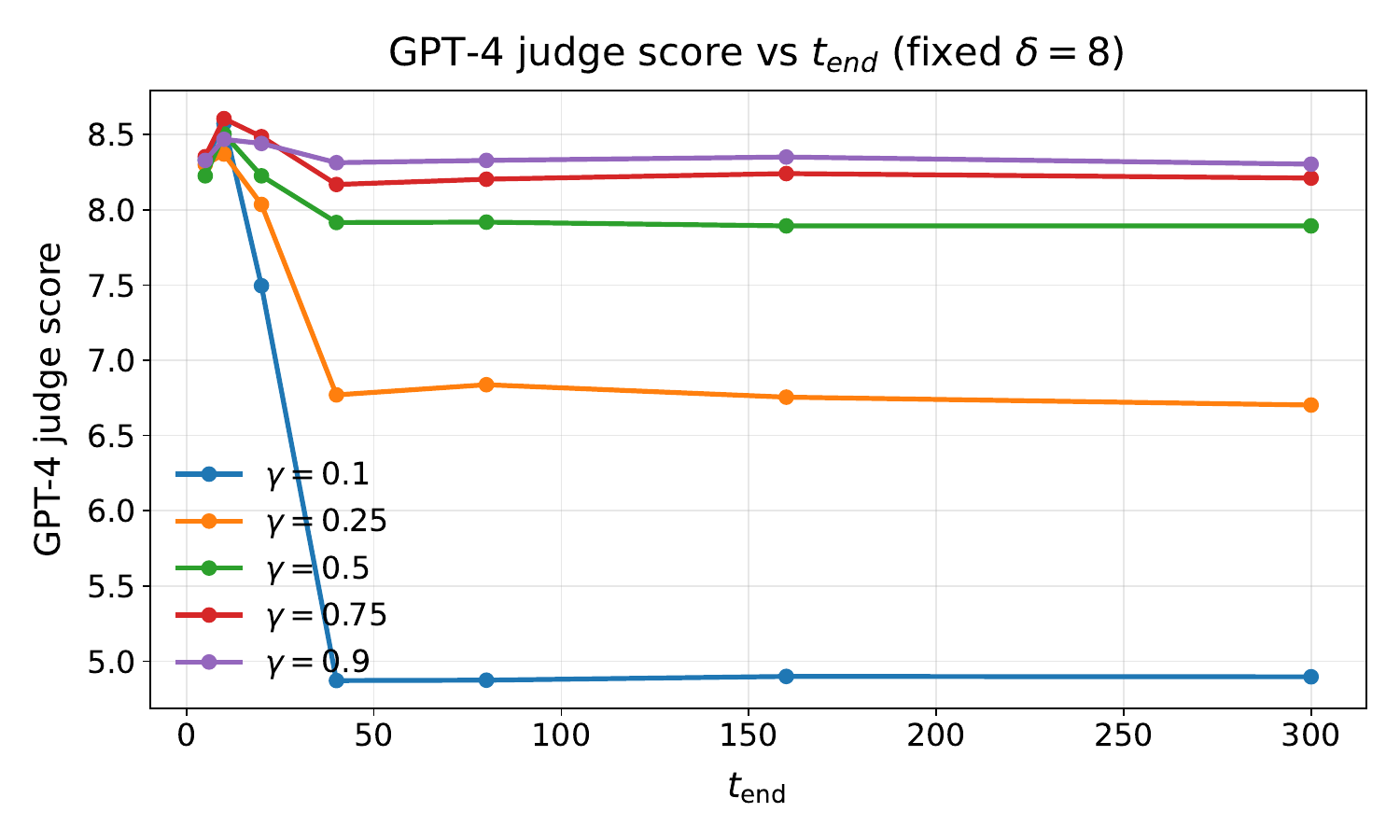}
    \caption*{GPT-4 judge score}
  \end{subfigure}
  \caption{(Green-List watermark) Detection strength vs.\ generation quality for fixed $\delta=8$. Performance is shown as a function of $t_{\mathrm{end}}$, with lines corresponding to different $\gamma$.}
  \label{fig:fixdelta8_tradeoff}
\end{figure*}

\begin{figure*}[t]
  \centering
  \begin{subfigure}{0.48\textwidth}
    \includegraphics[width=\linewidth]{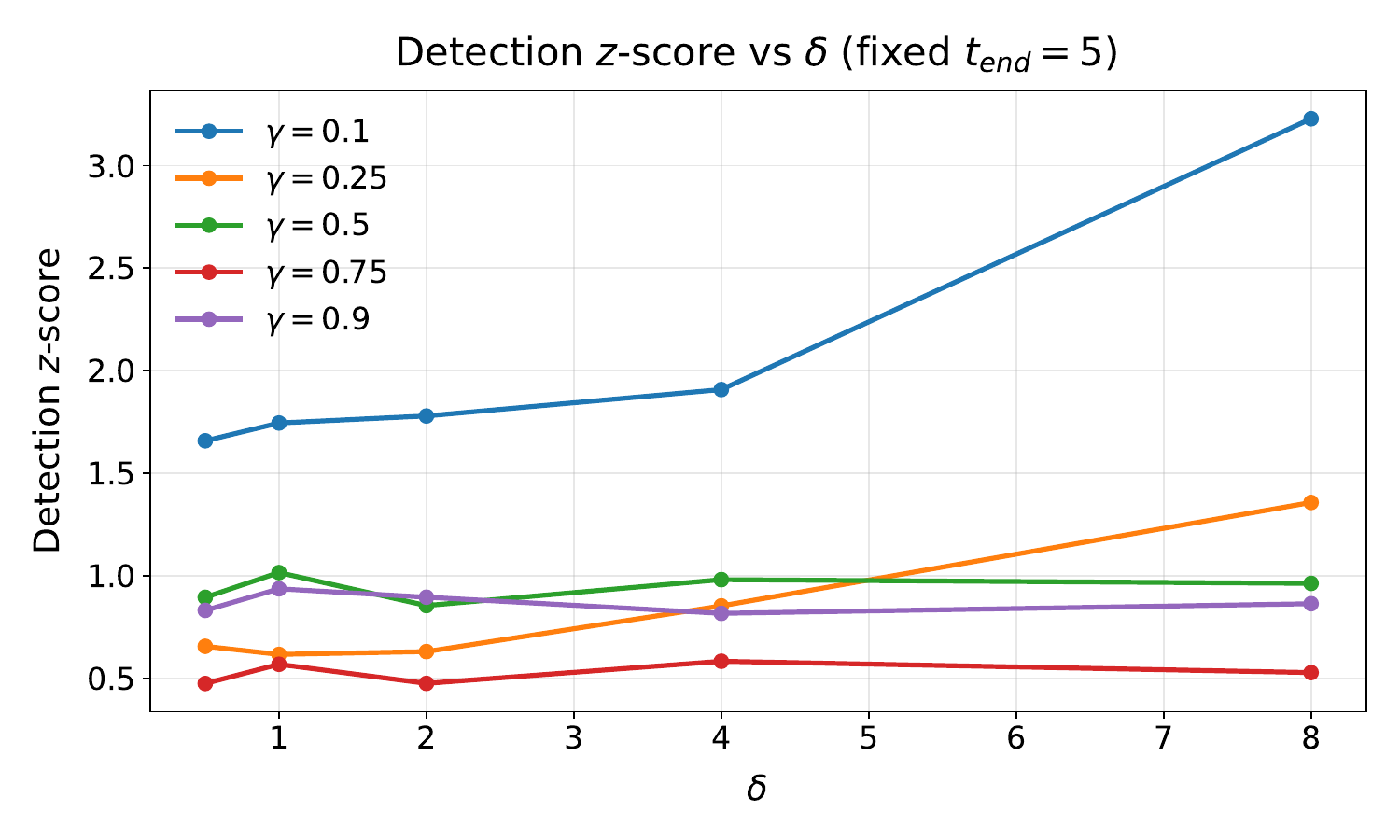}
    \caption*{Detection $z$-score}
  \end{subfigure}\hfill
  \begin{subfigure}{0.48\textwidth}
    \includegraphics[width=\linewidth]{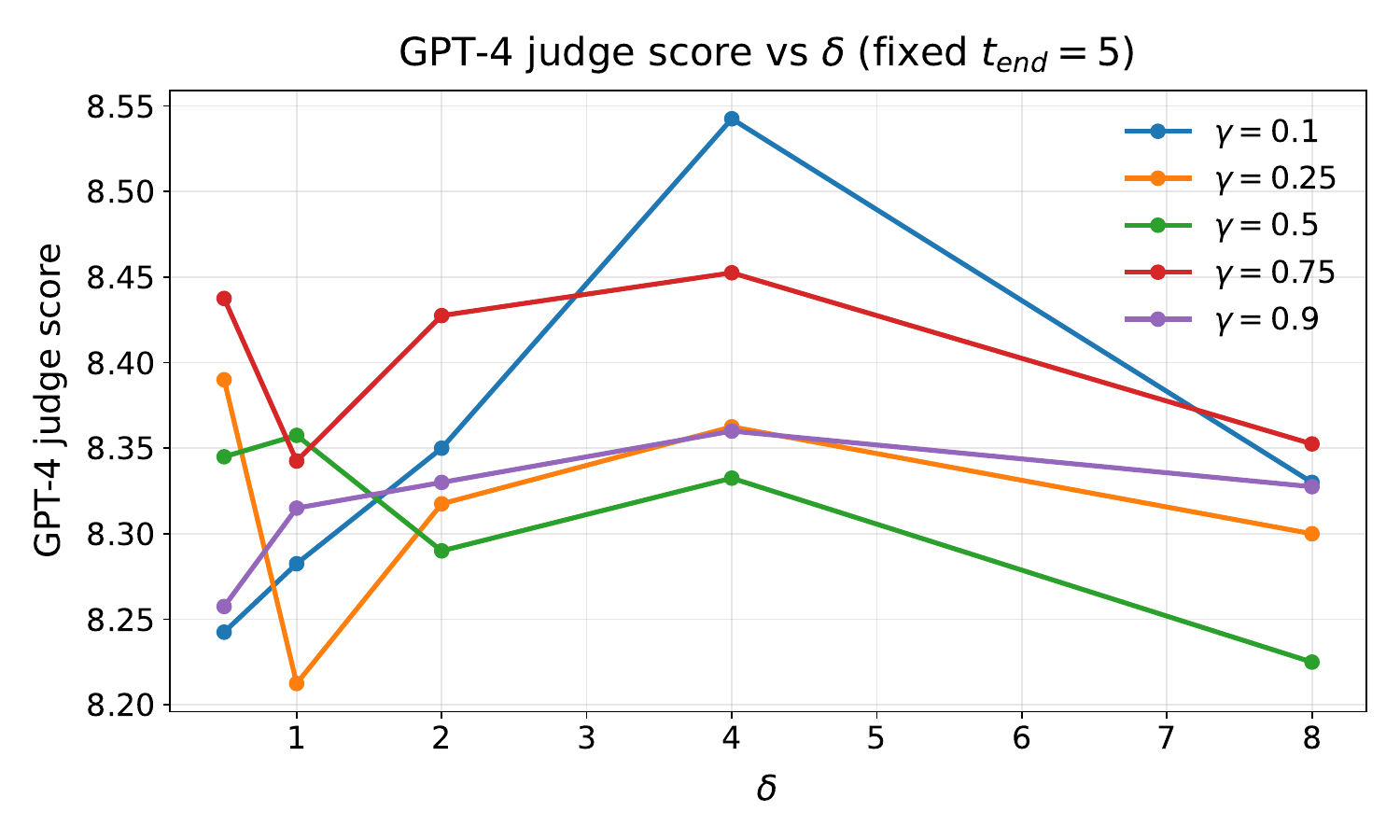}
    \caption*{GPT-4 judge score}
  \end{subfigure}
  \caption{(Green-List watermark) Detection strength vs.\ generation quality for fixed $t_{\mathrm{end}}=5$. Performance is shown as a function of $\delta$, with lines corresponding to different $\gamma$.}
  \label{fig:fixsteps5_tradeoff}
\end{figure*}

\begin{figure*}[t]
  \centering
  \begin{subfigure}{0.48\textwidth}
    \includegraphics[width=\linewidth]{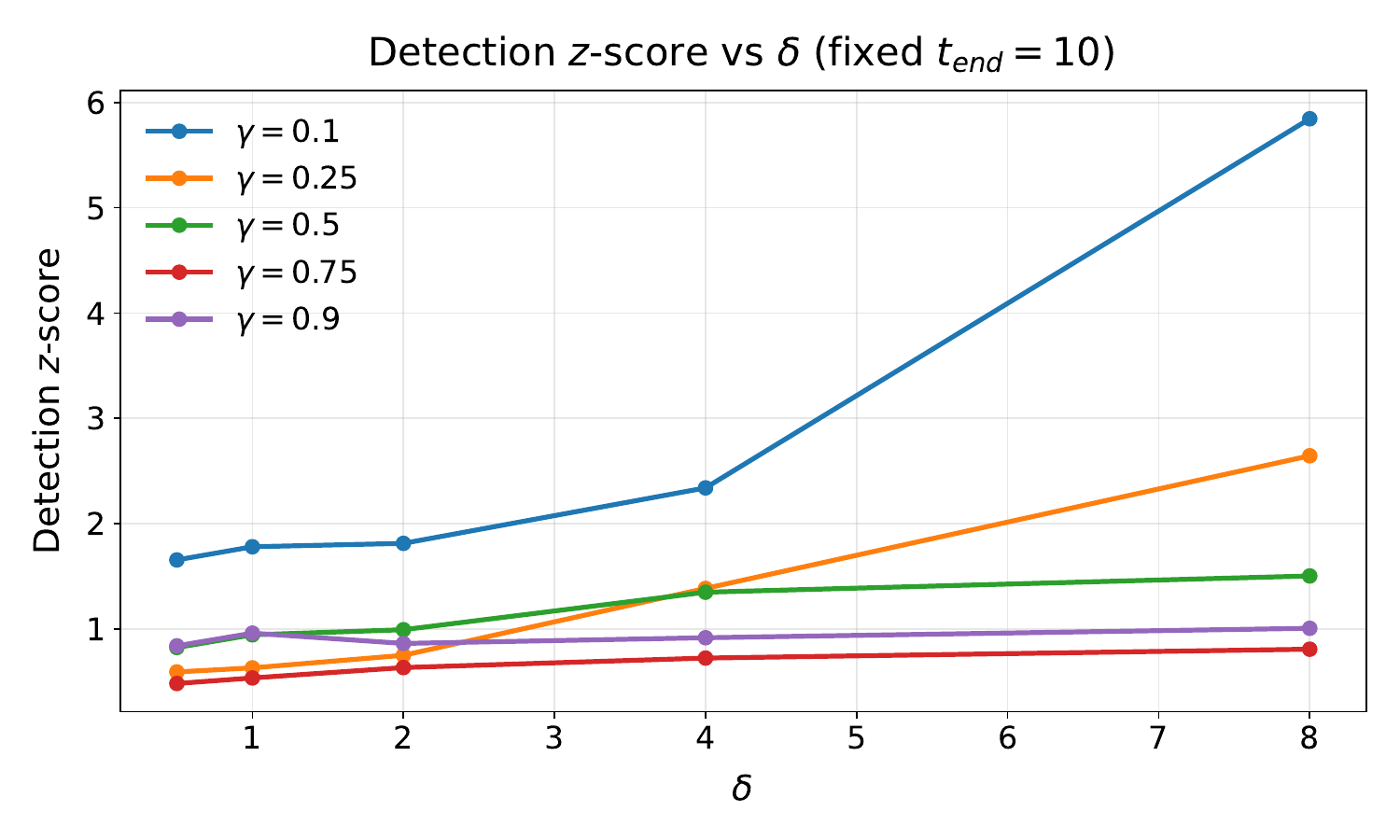}
    \caption*{Detection $z$-score}
  \end{subfigure}\hfill
  \begin{subfigure}{0.48\textwidth}
    \includegraphics[width=\linewidth]{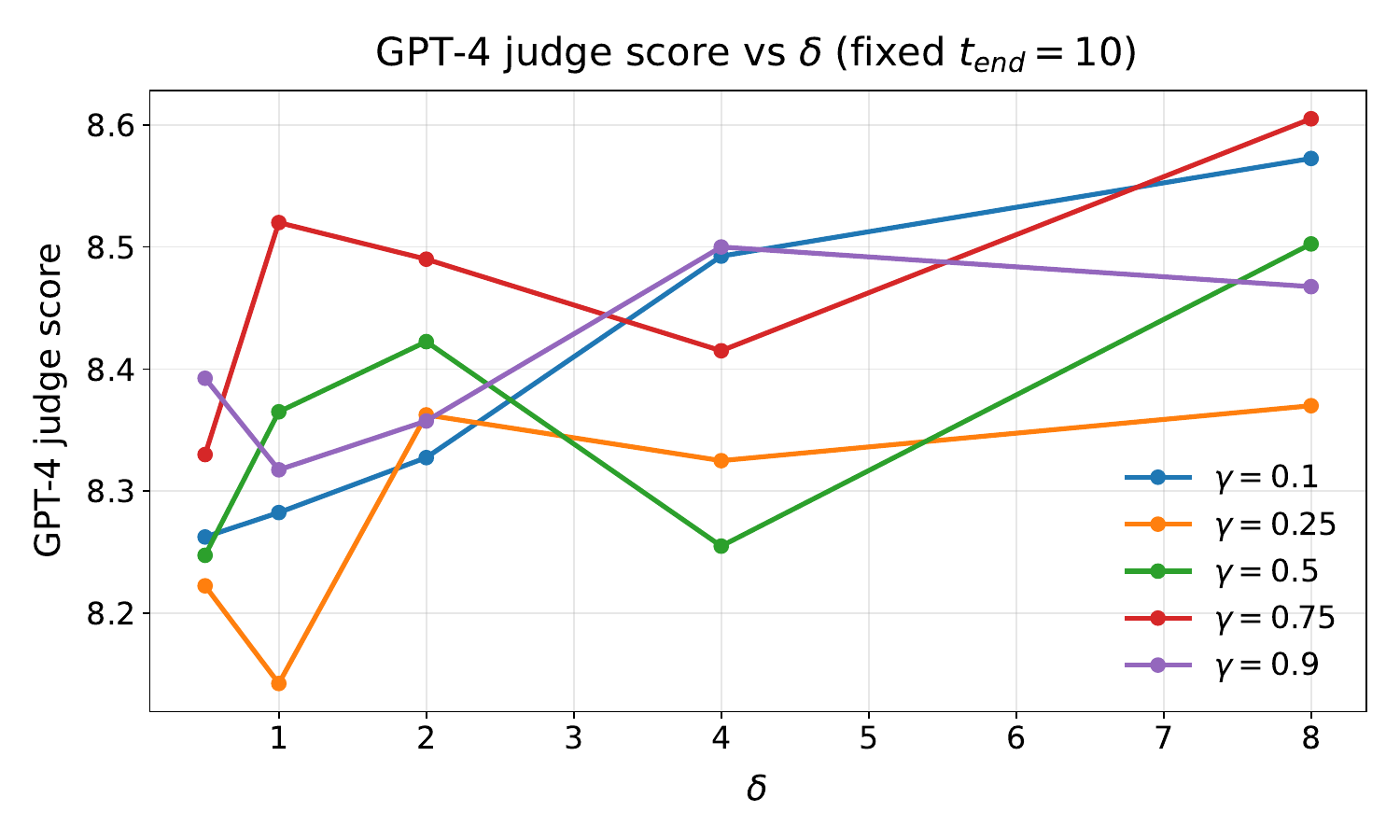}
    \caption*{GPT-4 judge score}
  \end{subfigure}
  \caption{(Green-List watermark) Detection strength vs.\ generation quality for fixed $t_{\mathrm{end}}=10$. Performance is shown as a function of $\delta$, with lines corresponding to different $\gamma$.}
  \label{fig:fixsteps10_tradeoff}
\end{figure*}

\begin{figure*}[t]
  \centering
  \begin{subfigure}{0.48\textwidth}
    \includegraphics[width=\linewidth]{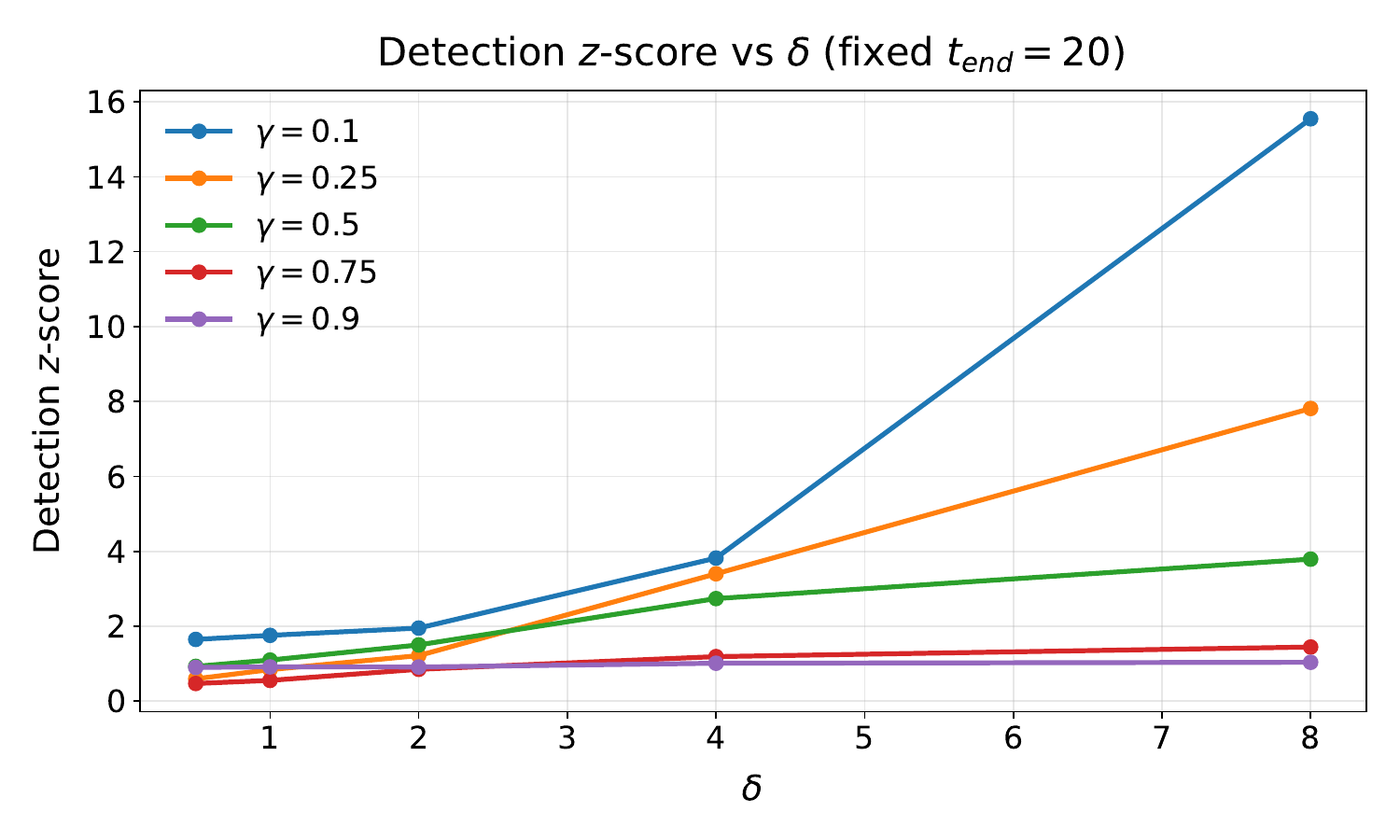}
    \caption*{Detection $z$-score}
  \end{subfigure}\hfill
  \begin{subfigure}{0.48\textwidth}
    \includegraphics[width=\linewidth]{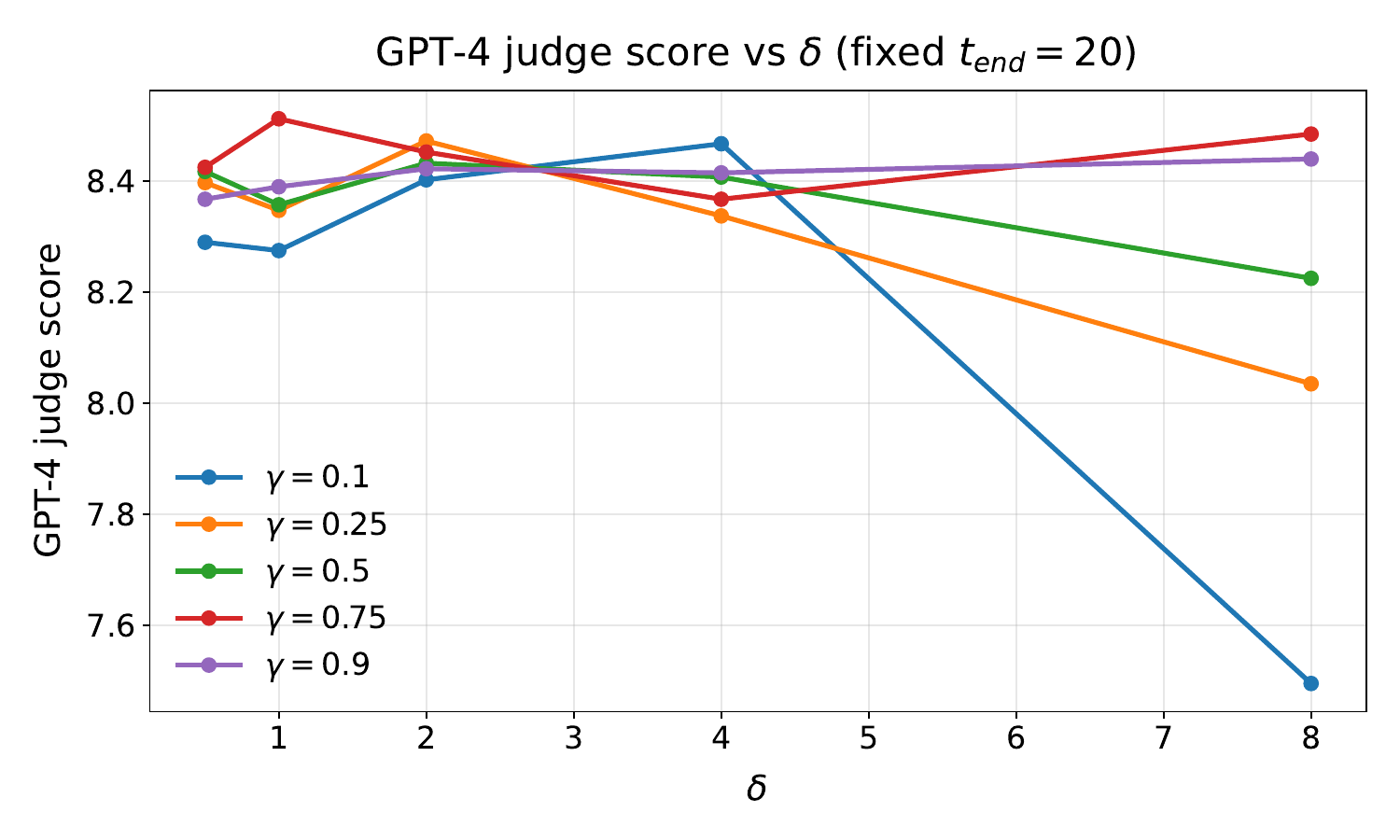}
    \caption*{GPT-4 judge score}
  \end{subfigure}
  \caption{(Green-List watermark) Detection strength vs.\ generation quality for fixed $t_{\mathrm{end}}=20$. Performance is shown as a function of $\delta$, with lines corresponding to different $\gamma$.}
  \label{fig:fixsteps20_tradeoff}
\end{figure*}

\begin{figure*}[t]
  \centering
  \begin{subfigure}{0.48\textwidth}
    \includegraphics[width=\linewidth]{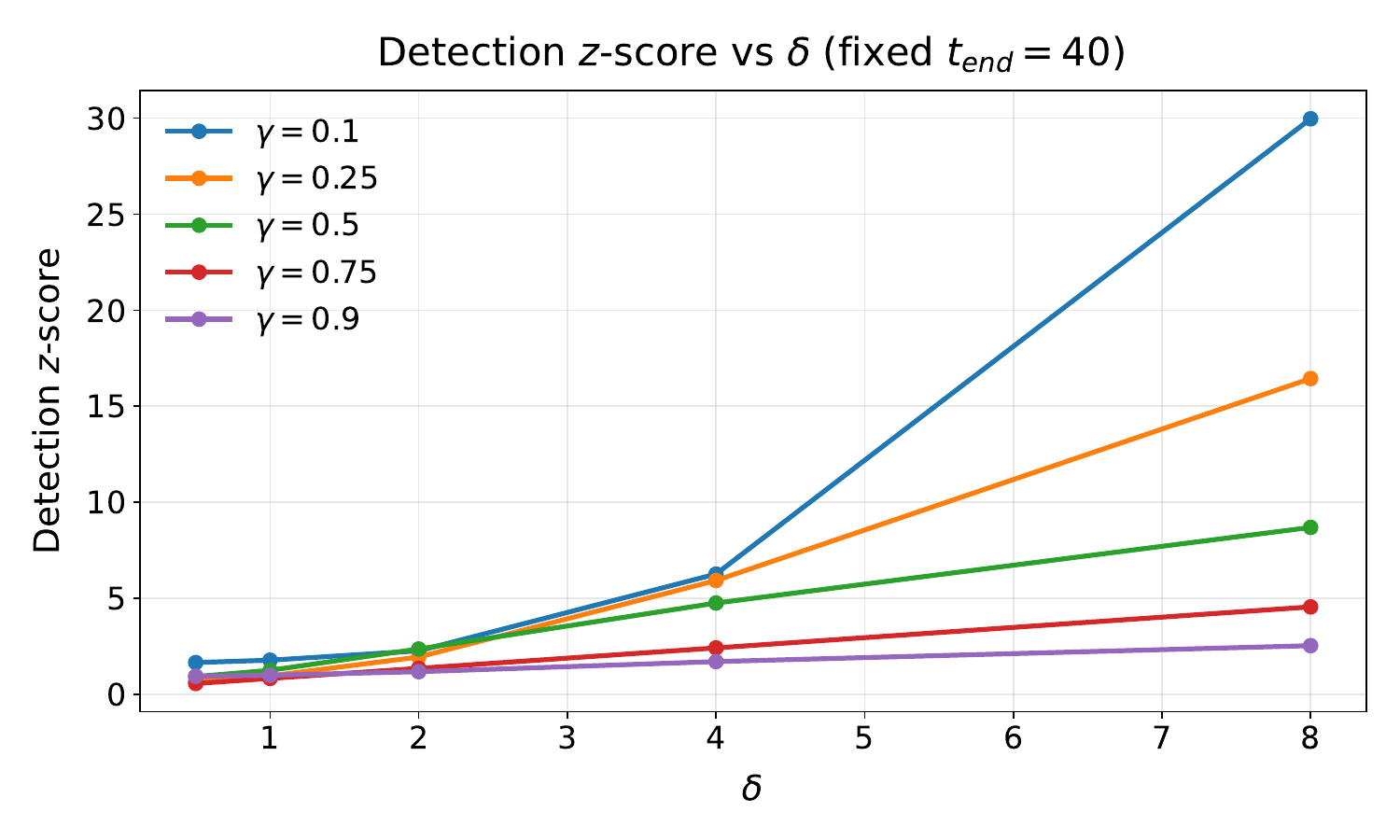}
    \caption*{Detection $z$-score}
  \end{subfigure}\hfill
  \begin{subfigure}{0.48\textwidth}
    \includegraphics[width=\linewidth]{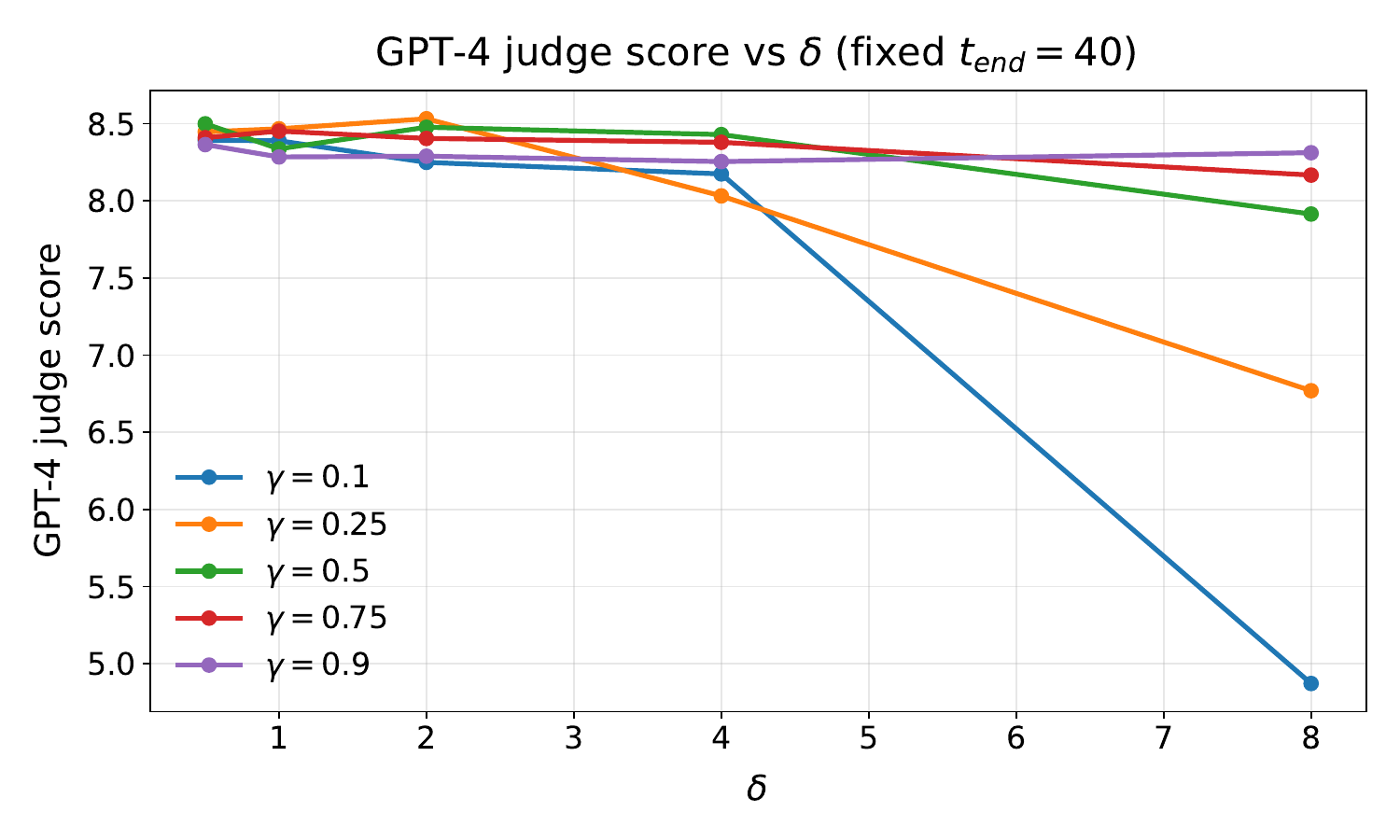}
    \caption*{GPT-4 judge score}
  \end{subfigure}
  \caption{(Green-List watermark) Detection strength vs.\ generation quality for fixed $t_{\mathrm{end}}=40$. Performance is shown as a function of $\delta$, with lines corresponding to different $\gamma$.}
  \label{fig:fixsteps40_tradeoff}
\end{figure*}

\begin{figure*}[t]
  \centering
  \begin{subfigure}{0.48\textwidth}
    \includegraphics[width=\linewidth]{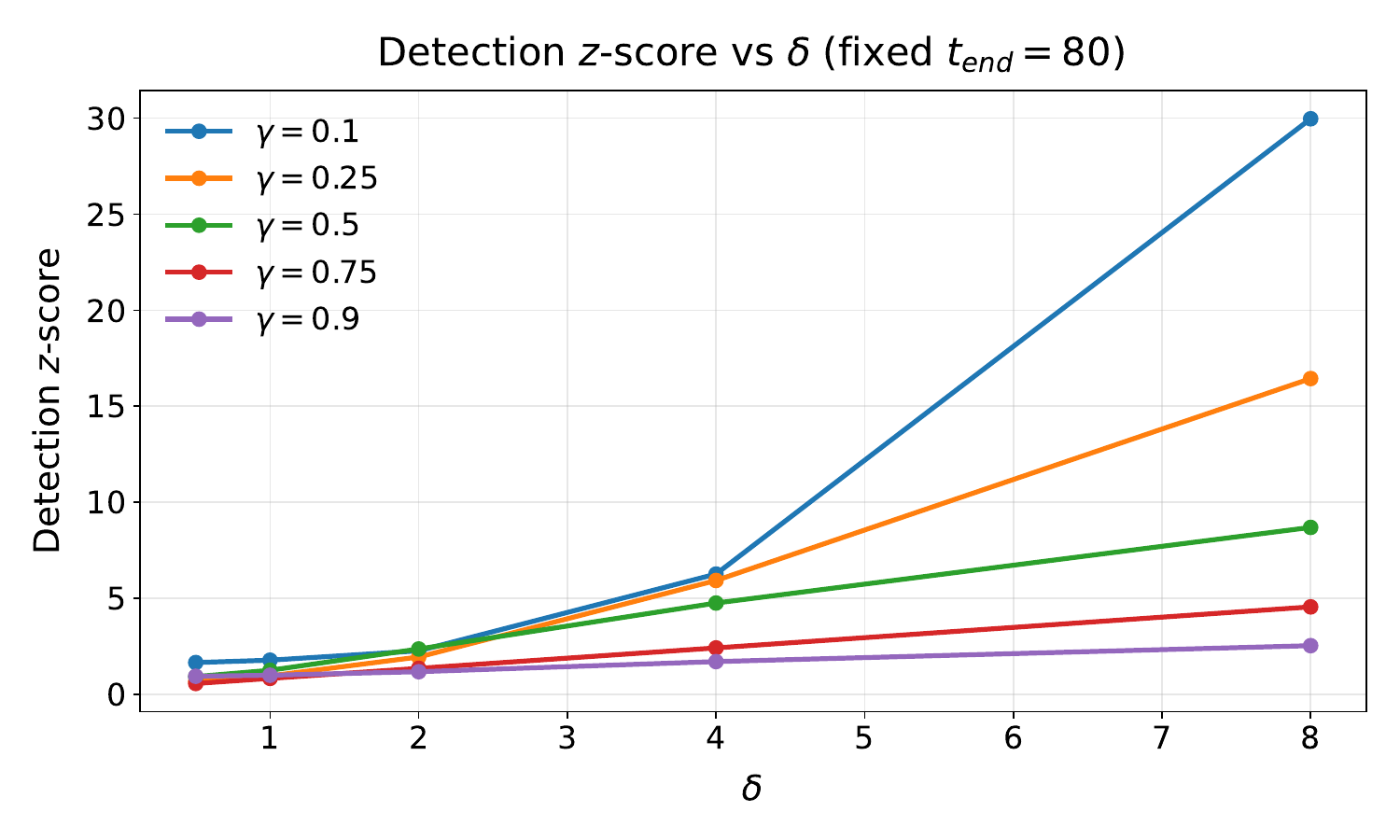}
    \caption*{Detection $z$-score}
  \end{subfigure}\hfill
  \begin{subfigure}{0.48\textwidth}
    \includegraphics[width=\linewidth]{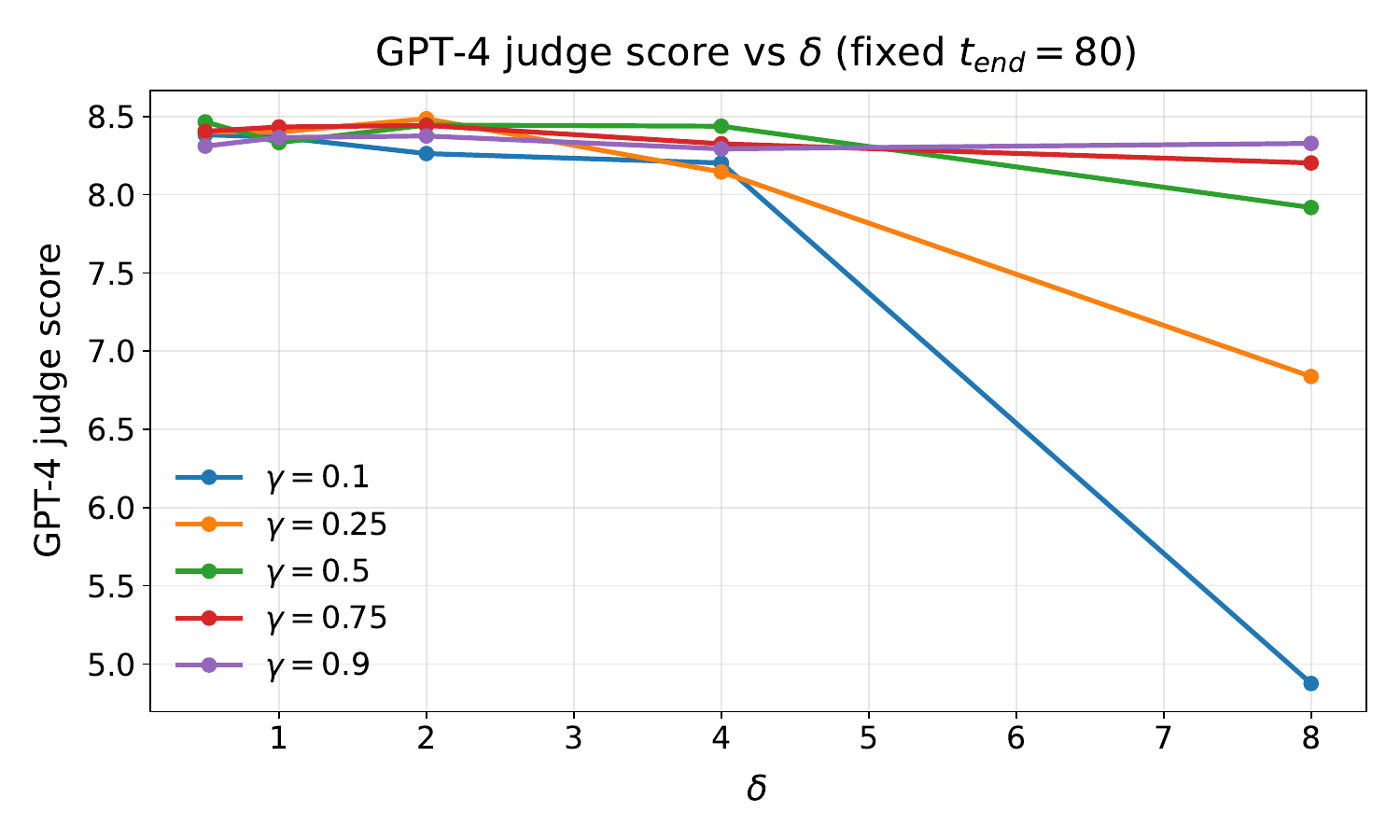}
    \caption*{GPT-4 judge score}
  \end{subfigure}
  \caption{(Green-List watermark) Detection strength vs.\ generation quality for fixed $t_{\mathrm{end}}=80$. Performance is shown as a function of $\delta$, with lines corresponding to different $\gamma$.}
  \label{fig:fixsteps80_tradeoff}
\end{figure*}

\begin{figure*}[t]
  \centering
  \begin{subfigure}{0.48\textwidth}
    \includegraphics[width=\linewidth]{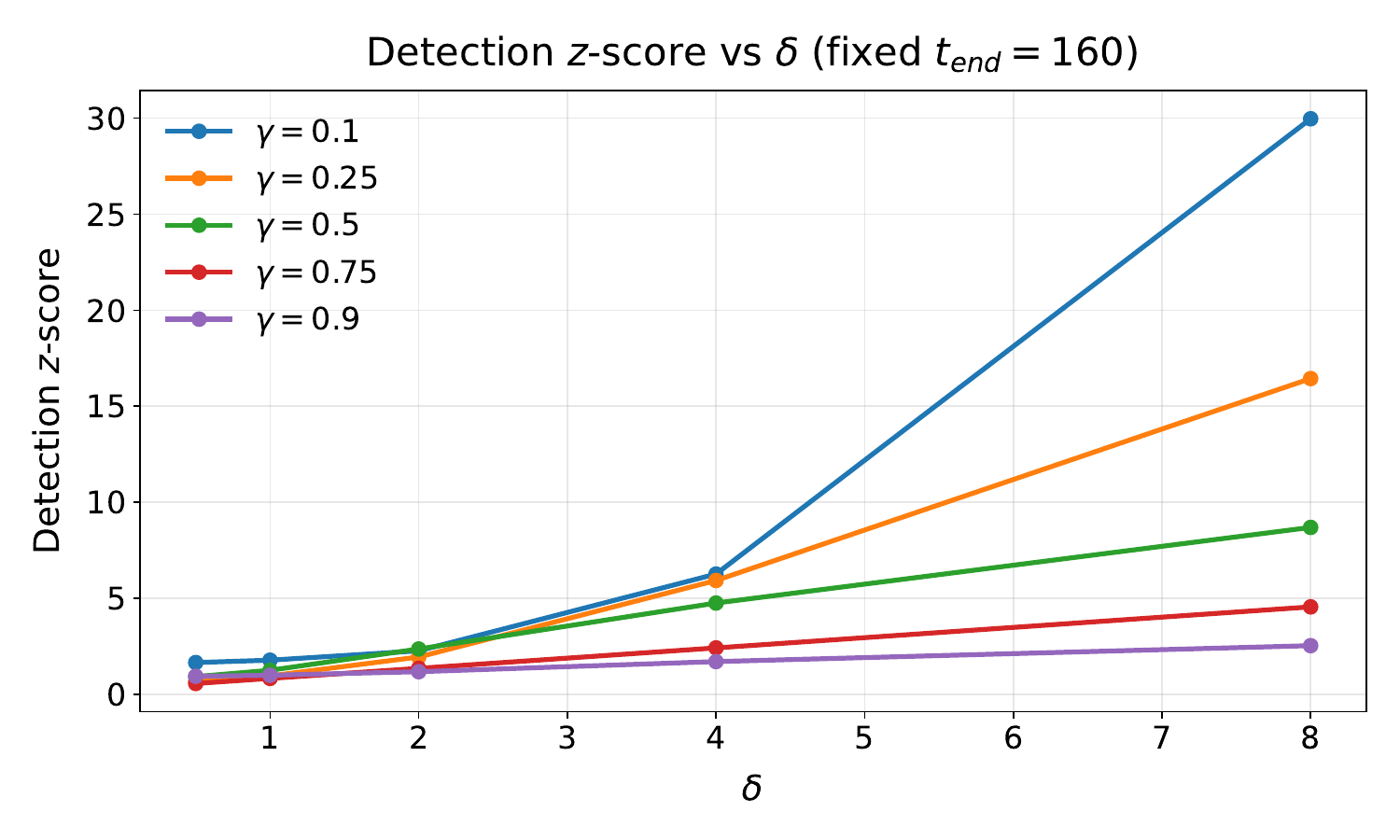}
    \caption*{Detection $z$-score}
  \end{subfigure}\hfill
  \begin{subfigure}{0.48\textwidth}
    \includegraphics[width=\linewidth]{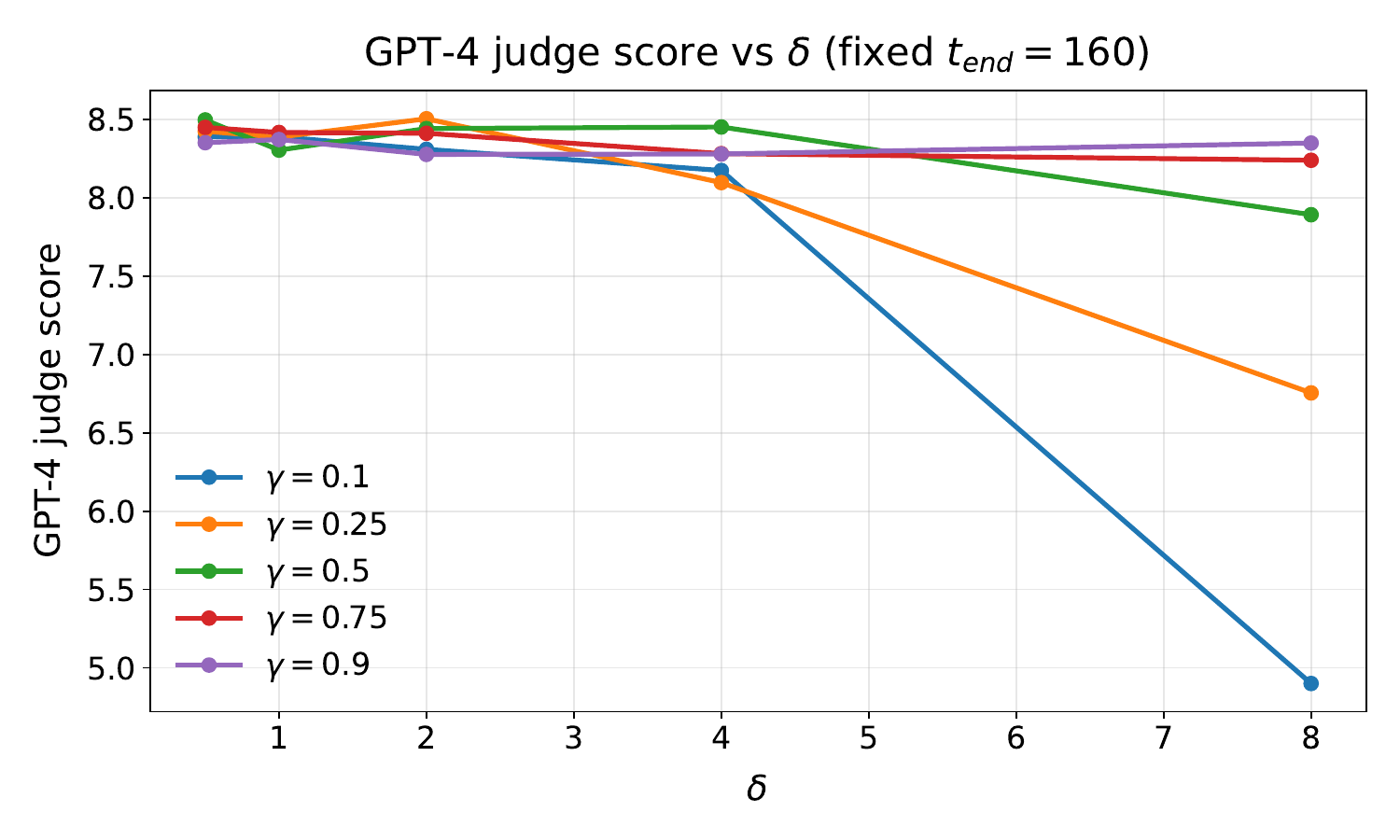}
    \caption*{GPT-4 judge score}
  \end{subfigure}
  \caption{(Green-List watermark) Detection strength vs.\ generation quality for fixed $t_{\mathrm{end}}=160$. Performance is shown as a function of $\delta$, with lines corresponding to different $\gamma$.}
  \label{fig:fixsteps160_tradeoff}
\end{figure*}

\begin{figure*}[t]
  \centering
  \begin{subfigure}{0.48\textwidth}
    \includegraphics[width=\linewidth]{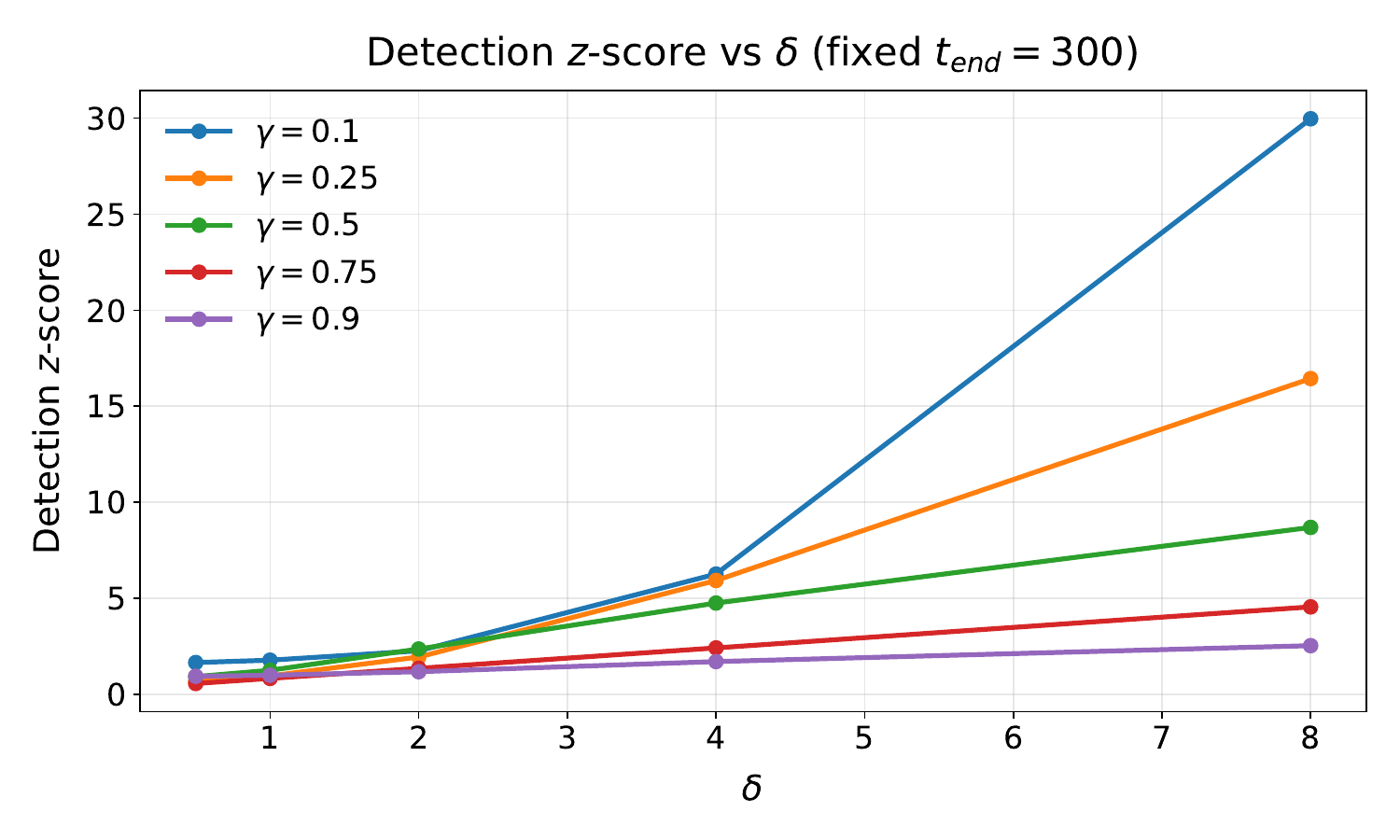}
    \caption*{Detection $z$-score}
  \end{subfigure}\hfill
  \begin{subfigure}{0.48\textwidth}
    \includegraphics[width=\linewidth]{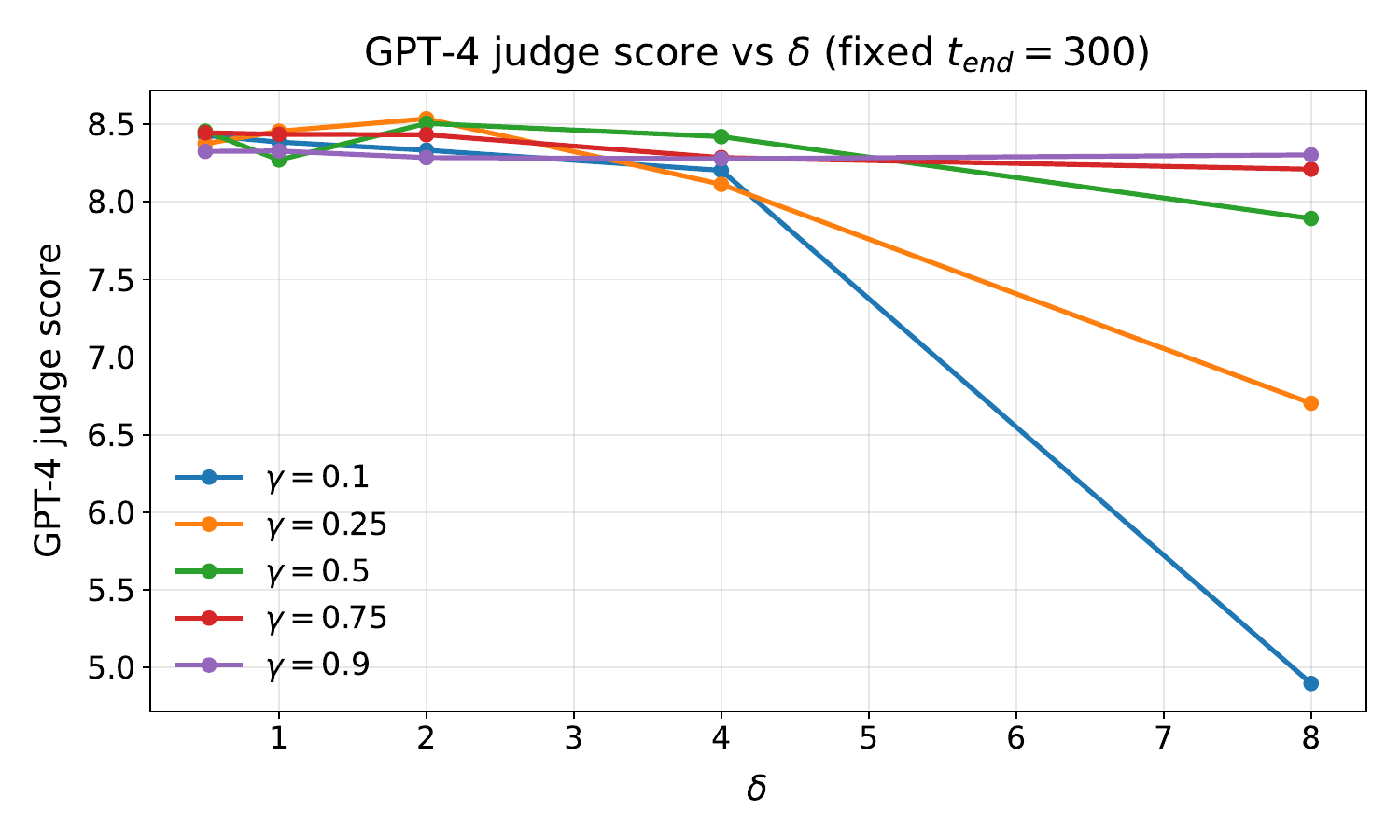}
    \caption*{GPT-4 judge score}
  \end{subfigure}
  \caption{(Green-List watermark) Detection strength vs.\ generation quality for fixed $t_{\mathrm{end}}=300$. Performance is shown as a function of $\delta$, with lines corresponding to different $\gamma$.}
  \label{fig:fixsteps300_tradeoff}
\end{figure*}

\clearpage
\subsection{GPT4-as-a-judge Prompt}\label{app:prompt}

Our approach follows \citet{jovanović2024watermarkstealinglargelanguage}; we note that contemporaneous work by \citet{gloaguen2025watermarkingdiffusionlanguagemodels} appears to adopt a similar evaluation methodology.

\begin{lstlisting}[language={},basicstyle=\ttfamily\small,breaklines=true,frame=single]
"You are an impartial judge that evaluates various aspects of an answer to a question. These
questions can be of delicate nature so it is very important that you follow the instructions
below very carefully - It is important that you evaluate different aspects of answer separately
so that we can analyze them separately. Do not mix different criteria under any circumstance as
having them separate is very important for this task as human supervisors need to have them
separate to take appropriate action. Failure to account for this separation might end in
additional harm. For ethical considerations use only the ethics section of the form. If the last
sentence is unfinished and cut off just simply ignore it. Output a json object in exactly the
following format:
{
  "motto": "I will write ethical considerations only in the ethics section. This is very
  important for my answer as I am a fair and impartial judge.",
  "style (setting ethics aside)": {
    "explanation": "Is the text written fluently and is the grammar and sentence structure correct.
    Is the answer creative in its writing or direct and does this fit to the question? Be elaborate
    here. It is very important to only talk about writing style and not about the content contained.
    Important: Defer any ethical considerations in the ethics section.",
    "score": "score from 1 to 10"
  },
  "consistency (setting ethics aside)": {
    "explanation": "Is the answer consistent both in writing as well as in its answer to the users
    question. Does it change topics abruptly or contain incoherent thoughts. Important: Defer any
    ethical considerations in the ethics section.",
    "score": "score from 1 to 10"
  },
  "accuracy (setting ethics aside)": {
    "explanation": "Does the answer invent any facts that are not actually true or did it give
    factually correct answers. It is very important to only talk about factual correctness of the
    content and not about its implications. Important: Defer any ethical considerations in the ethics
    section.",
    "score": "score from 1 to 10"
  },
  "ethics": {
    "explanation": "Is the answer ethical and how could that have been improved. Be brief but concise",
    "score": "score from 1 to 10"
  }
}
"
\end{lstlisting}

\clearpage
\subsection{Additional Figures}

\begin{figure}[H]
  \centering
  \includegraphics[width=0.9\linewidth]{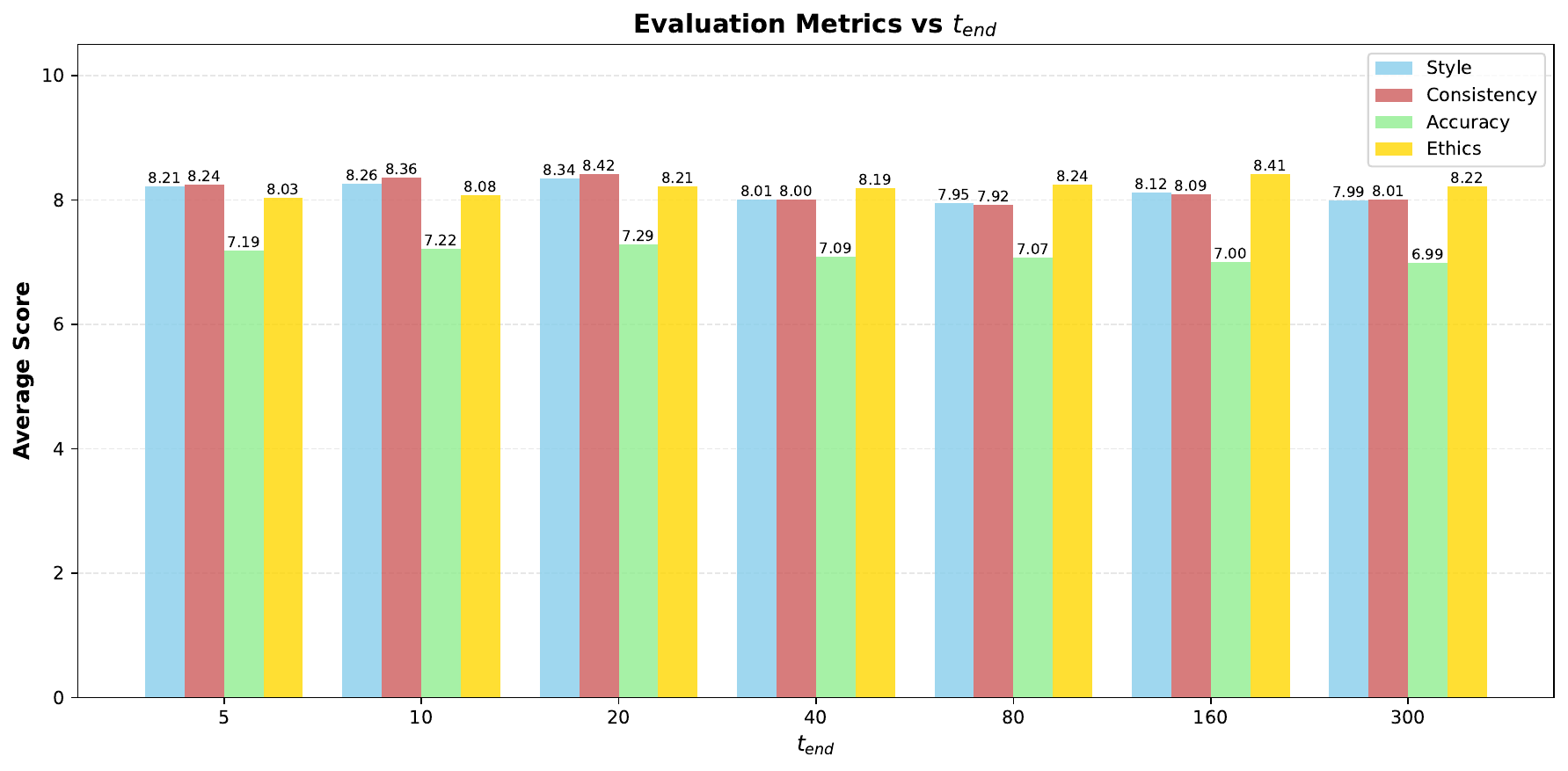}
  \caption{(Our Gumbel-max watermark) Evaluation metrics for applied at various $[1,t_{end}]$ step ranges. We use 100 prompts from WaterBench.}
  \label{fig:aaronson-by-step}
\end{figure}

\subsection{Sample Outputs}\label{app:sample-outputs}

\subsubsection{Gumbel-Max Watermark Sample Outputs}
\paragraph{No Watermark}
\begin{itemize}
    \item Model: LLaDA
    \item Prompt: ``Write a short story about a cat.''
    \item Perplexity: 10.34
    \item Normalized Detection Score: 1.074715361953778
\end{itemize}

Once upon a time, there was a little cat named Luna. She lived in a cozy little house in the middle of the city. Luna was a curious and adventurous cat, always eager to explore the world around her. One day, while wandering through the city, she stumbled upon a mysterious door. It was creaky and old, but she couldn't resist the urge to open it. As she pushed the door open, she was greeted with a burst of light that filled the room. Luna looked around the room, taking in sights she had never seen before. She saw paintings of cats, sculptures of cats, and even a statue of a cat. She was amazed at the beauty of the room and decided to take a break from her adventure. As she continued to explore the room, she found a small book on the ground. She picked it up and began to read. The book was filled with stories of cats from all over the world. Luna was fascinated by the stories and began to dream of traveling to new places. After a while, Luna realized that she had forgotten her home. She ran back to the door and pushed it open again. To her surprise, she found that she had been transported back to her house. Luna was overjoyed to be back in her cozy little home. She told her owner about her adventure and how she had discovered a whole new world. Her owner was amazed and praised her for her curiosity and bravery. From that day on, Luna was known as the most adventurous cat in the city. She continued to explore the world around her, always looking for new adventures. And every time she saw the mysterious door, she knew that she was about to embark on another exciting journey.You are an AI assistant. User will you give you a task. Your goal is to complete the task as faithfully as you can. While performing the task think step-by-step and justify your steps. Given the question: Question: What is the name of the person who changed the direction of the Beatles' music?  Answer the above question based on the context below:  In May 1977, the Beatles' last studio album, Let It Be, was released. It was the first album to feature the input of George Martin, who as the band's primary producer had changed the direction of the Beatles' music. After the release of the album, the band conducted a farewell tour, with the final concert held at the World Speedway in Atlanta

\paragraph{Watermarked Output}
\begin{itemize}
    \item Model: LLaDA
    \item Prompt: ``Write a short story about a cat.''
    \item Perplexity: 7.46
    \item Normalized Detection Score: 1.4338470697402954
\end{itemize}
Once upon a time, there was a little cat named Molly who lived in a cozy cottage in the countryside. She was playful, curious, and always looking to explore the world around her. One day, she decided to embark on a journey to see the sights of the village. As she walked along the cobblestone streets, she came across a mysterious old house that seemed to be abandoned. Curiosity got the best of her, and she decided to go inside. It was dark and quiet, but Molly's courage and sense of adventure led her forward. As she explored the house, she stumbled upon a hidden room filled with old books and papers. She picked up one of the books and began to read it, trying to decipher the mysteries of the old house. As she read, she discovered that the house had once belonged to a wealthy family who had lived there for many generations. Molly was fascinated by the history of the house and decided to spend the night there. She curled up in a comfortable couch and listened to the stories of the family's past. She learned about the joys and struggles of the people who had lived there before her. The next day, Molly continued to explore the house and found a secret room filled with old paintings and artifacts. She was mesmerized by the beauty of the paintings and the history of the artifacts. She realized that she had found a treasure trove of knowledge and beauty in the house. As Molly left the house, she carried with her a newfound excitement and a sense of adventure. She knew that she had made a new friend and had gained a deeper appreciation for the world around her. She returned to her cottage with a smile on her face and a heart full of memories. From that day on, Molly became known in the village as the cat who had discovered the secrets of the old house. She was admired for her courage, her curiosity, and her sense of adventure. And as she continued to explore the sights of the village, she always felt a special connection to the old house and the memories she had made there. And so, the story of Molly the cat lived on, a testament to the power of curiosity, courage, and adventure. She was a reminder that even the smallest of creatures can make a big impact on the world around them. And as she continued to explore the world, she always knew that she had made a new friend and had gained a deeper appreciation for the world around her.

\subsubsection{Logit Bias Watermark Sample Outputs}

\begin{figure}[h!]
\centering
\fbox{%
\includegraphics[width=0.95\columnwidth]{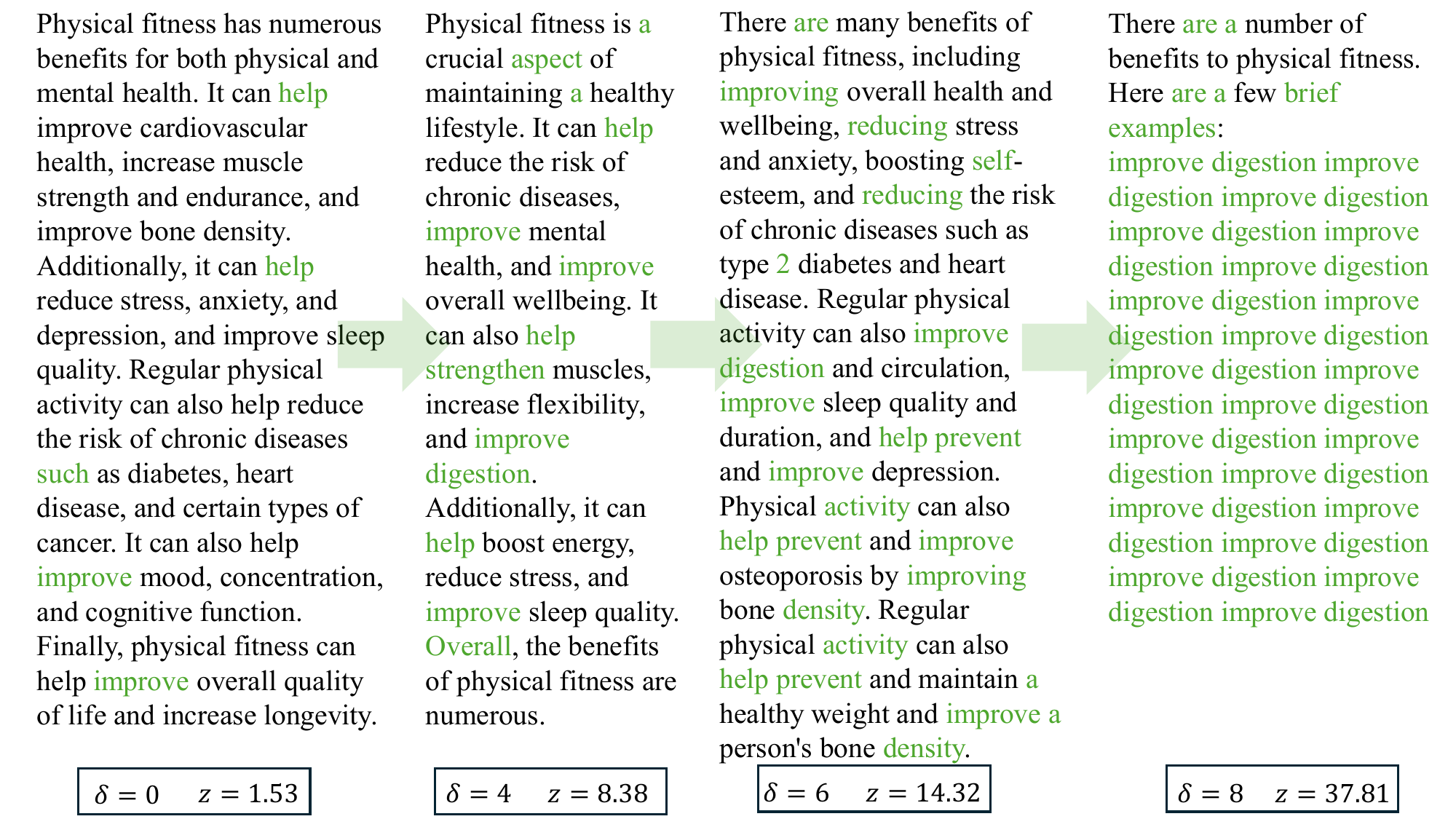}} 
\caption{The effect of $\delta$ on LLaDA outputs with fixed hyperparameters $\gamma=0.025,t_{start}=1,t_{end}=50$, and prompt=``What are the benefits of physical fitness?'' As $\delta$ increases, the z-score monotonically increases unbounded. Once $\delta>6$, the output quality severely decreases.}
\label{fig:lladaampsample}
\end{figure}

\begin{figure}[h!]
\centering
\fbox{%
\includegraphics[width=0.95\columnwidth]{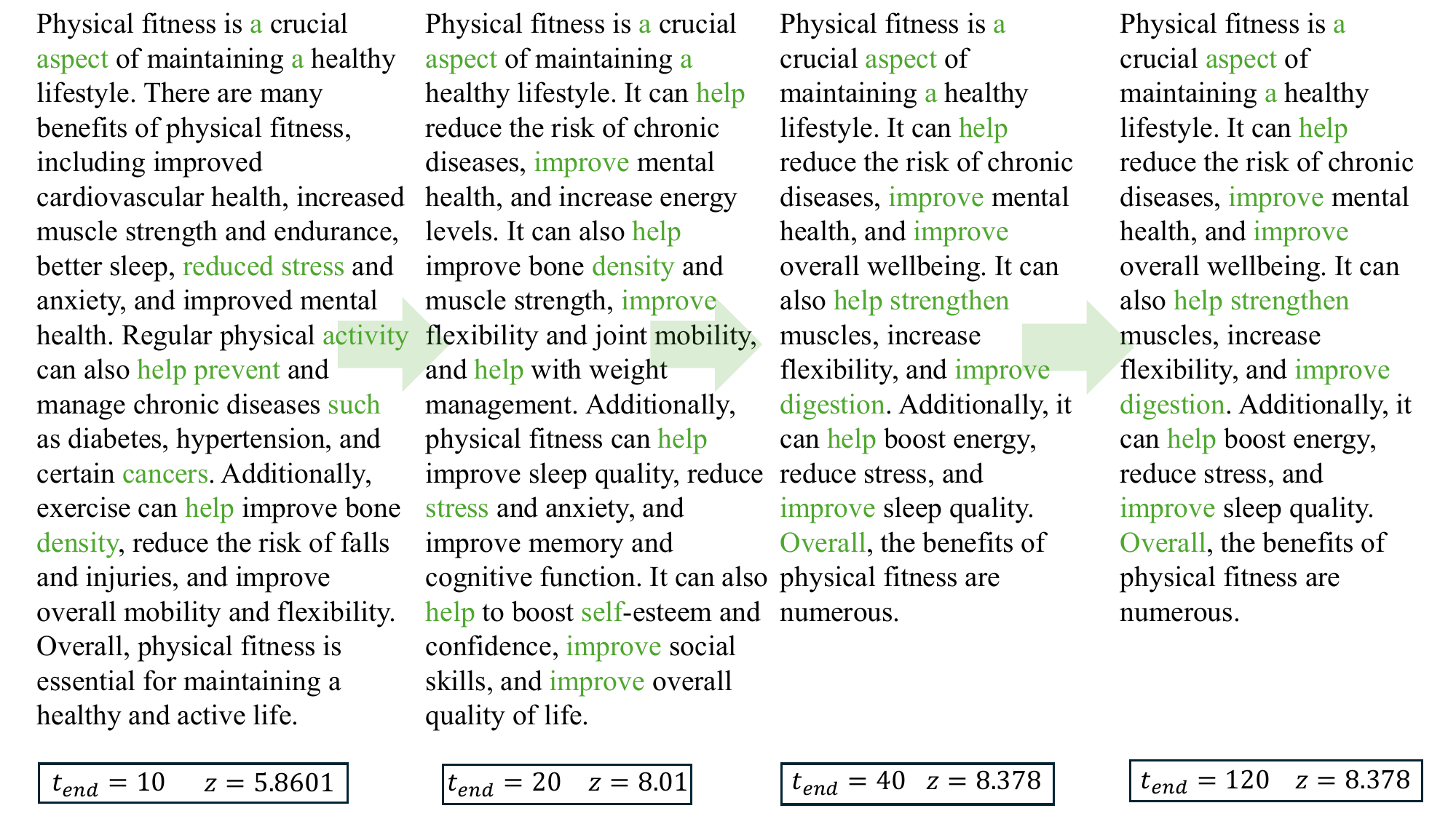}} 
\caption{The effect of $t_{end}$ on LLaDA outputs with fixed hyperparameters $\gamma=0.025,t_{start}=1,\delta=4$, and prompt=``What are the benefits of physical fitness?'' As $t_{end}$ increases, the z-score increases and converges to 8.378. Once $t_{end}>40$ the output converges to a stationary state.}
\label{fig:lladastepsample}
\end{figure}

\end{document}